\documentclass[journal]{IEEEtran}

\usepackage[pdftex]{graphicx}
\graphicspath{{../pdf/}{../jpeg/}}
\DeclareGraphicsExtensions{.pdf,.jpeg,.png}
\usepackage{xcolor,soul,framed}
\colorlet{shadecolor}{yellow}

\usepackage{amsmath}
\usepackage{amssymb}
\usepackage{amsthm}
\usepackage{mathtools}

\DeclareMathOperator*{\argmin}{arg\,min}

\newtheorem{theorem}{Theorem}
\newtheorem{proposition}{Proposition} 

\theoremstyle{definition}
\newtheorem{definition}{Definition}
\newtheorem{remark}{Remark}

\usepackage{array, mdwmath, mdwtab, eqparbox, multirow, float}
\usepackage{booktabs}
\usepackage{multicol}
\usepackage[bottom]{footmisc} 

\usepackage{mathptmx}
\usepackage{helvet}
\usepackage{courier}
\usepackage{type1cm}

\usepackage[ruled]{algorithm}
\usepackage{algpseudocode}

\usepackage{makeidx}
\usepackage{cite} 
\begin{document}
\title{ Distributed Deep Variational Approach for Privacy-preserving Data Release
    }
   \author{Zahir Alsulaimawi,~\IEEEmembership{ Member,~IEEE,}
      Huaping Liu,~\IEEEmembership{Member,~IEEE}

}


\maketitle

\begin{abstract}
Federated learning (FL) lets distributed nodes train a shared model without exchanging their raw data, but in privacy-sensitive deployments medical sensors, IoT devices, wearables the protection offered by keeping data local is incomplete: gradients, model updates, and the released representations themselves can leak sensitive attributes. We propose the \emph{Gaussian Privacy Protector} (GPP), a data-release framework for continuous, high-dimensional inputs that learns a stochastic encoder mapping raw data to a low-dimensional sanitized representation. The encoder is trained against a variational lower bound on the mutual information between the released representation and a designated sensitive attribute, while a separate cross-entropy term preserves a designated utility attribute, with a Lagrange multiplier $\beta$ controlling the trade-off. We then extend GPP to the federated setting, in which each client trains a local encoder, sensitive labels never leave the client, and the aggregator receives only sanitized representations giving instance-level privacy protection in addition to the standard ``raw data stays local'' guarantee of FL. We evaluate GPP on MNIST (digit-sum utility, parity sensitive), CelebA (smiling vs.\ gender), and HAPT-Recognition (activity vs.\ subject identity). Across all three benchmarks, GPP attains utility within roughly one percentage point of an unconstrained autoencoder baseline while reducing the adversary's AUC to near random guessing.
\end{abstract}


%
\IEEEpeerreviewmaketitle

\section{Introduction}

\IEEEPARstart{W}{ith}  today's privacy-preserving techniques, it is critical to ensure that the data are able to convey useful information without disclosing sensitive information.  
Taking a small step toward achieving privacy could be as simple as removing sensitive information from a dataset. We may use the example of removing race-related information from an image while keeping gender as public information. However, information regarding utility data can be lost due to the correlation between sensitive and valuable information. 
Several privacy-preserving computation strategies are available for machine learning, including anonymization and differential privacy, both of which perturb the data to some extent \cite{wagner2018technical}.
In contrast to these perturbation techniques, which provide privacy guarantees for categorical attributes, they may not be appropriate for continuous high-dimensional features such as images, videos, and audio clips \cite{narayanan2008robust, sarwate2013signal}. 
A model-inversion attack can also be used to reconstruct a part of the training data using only predictions, as when recovering facial recognition images \cite{fredrikson2015model}. In addition, by observing only the model's predicted values, a membership-inference technique can determine if a particular training point appears in the model's training data \cite{shokri2017membership}.

As part of our work, we address the issue of privacy-preserving data release, which aims to minimize exposure of sensitive information associated with data while still allowing useful data to be released.
This paper presents a framework for mapping continuous data from a high-dimensional space into a sanitized format in another compressed space. 
We examine a scenario where users have two types of data: sensitive data that they wish to keep private and utility data (which is not sensitive) that they wish to make publicly available. 
Besides providing as much utility information as possible, the new representation conceals any sensitive information regarding private events.
We call this approach Gaussian privacy protector (GPP). The objective of GPP is to train a probabilistic map, an encoder, along with an adversarial network,   classifier(s), to recover private information from a sanitized dataset, and with a utility network,   classifier(s), to obtain utility information. A privacy-utility accuracy is achieved by utilizing an information-theoretic approach for private and non-private data and analyzing the cross-entropy (CE) function.


According to ubiquitous data collection, individuals constantly produce diverse swaths of data, including location, health, and financial information. 
These data streams are often obtained by distributed learning techniques such as the internet of things (IoT), ubiquitous sensing, edge computing, and many other distributed systems \cite{ma2021uav}.
Federated learning (FL) enables the training of collaborative models over a large number of participated IoT applications using a global server.  
However, FL may be exploited by malicious participants through backdoor attacks. Also, the leakage of data, gradients, and even models during the updating and transmitting process has raised user privacy and security concerns, limiting its application \cite{liu2021privacy}, \cite{lyu2020towards}.
The paper proposes a distributed system to address the issues raised by FL environments.  The privacy problem can be framed as a distributed computing environment where multiple GPPs train sensitive data locally and share sanitized data with an authorized aggregator (server).

\noindent\textbf{Contributions:}
Our contributions lie not in inventing the individual ingredients---variational mutual-information bounds, adversarial censoring, and Gaussian latent spaces all appear in prior work that we cite explicitly in Section~\ref{sec:majhead}---but in (i) integrating them into a single end-to-end privacy-utility objective with a transparent variational derivation, and (ii) extending the framework to a distributed setting in which each client retains its sensitive labels and a local adversary. With that scoping in mind, the specific contributions of this paper are as follows.

\begin{itemize}
\item We give a self-contained variational derivation of an end-to-end privacy-utility objective that combines a lower bound on $I(Z;S)$ with an upper bound on $H(U \mid Z)$ in a single saddle-point form, and present GPP, a Gaussian-latent encoder trained against this objective with the trade-off exposed as a single tunable parameter $\beta$.

\item We extend GPP to a distributed setting (Distributed GPP) in which each client trains a local encoder and a local adversary, and the aggregator receives only sanitized representations and utility labels. This architectural split enables instance-level privacy protection that complements the standard ``raw data stays local'' guarantee of federated learning, and is the principal contribution of this paper.

\item We prove a compositional privacy bound (Proposition~\ref{prop:distributed_privacy}) for the distributed setting under IID-data and honest-aggregator assumptions, with an explicit residual term $\delta$ accounting for utility-label leakage and an explicit scoping remark identifying the threats the bound does and does not address.

\item We characterize the communication-cost trade-off of Distributed GPP relative to gradient-based federated learning, identifying the regime $b(d_z + d_u) < |\theta|$ in which sanitized-representation transmission strictly reduces per-round payload.

\item We evaluate GPP on three benchmarks (MNIST, CelebA, HAPT-Recognition) covering binary, multi-class, and identity-style sensitive attributes, and conduct ablations on $\beta$, $d_z$, the adversary update frequency $k$, and the utility-sensitive correlation $\rho$.
\end{itemize}

The rest of the paper is organized in the following manner. The related work is presented in Section II, while the preliminary work is discussed in Section III. In Section IV, the problem formulation, Bayesian network, and threat model are formalized. The details of the GPP's algorithm are explained in Section V. The distributed learning algorithm is introduced in Section VI. In section VII, the GPP framework is  empirically examined. The paper concludes with Section VIII. 

\section{Related Work}
\label{sec:format}

Several generic privacy-preserving  models have been proposed to protect data privacy by increasing the amount of uncertainty, for instance, k-anonymity  \cite{sweeney2002k}, l-diversity \cite{machanavajjhala2006diversity}, and t-closeness \cite{li2007t}.
These approaches are only suitable for low-dimensional data because quasi-identifiers and sensitive attributes cannot be easily defined for high-dimensional data.
The most prominent scheme, however, is given by differential privacy (DP) \cite{dwork2014algorithmic, oneto2017differential}. DP is a mathematical analysis leading to a strong definition of privacy on a statistical basis.  By adding well-designed noise to a database, DP creates a more formal method for open-sourcing a database and keeping individual records private.
Privacy guarantees that require more noise in the data often limit its application scenarios, especially when high accuracy of learning tasks are required.
Additionally, privacy guarantees may be compromised by ignoring the data distribution in the methods described above. As exemplified by \cite{kifer2011no} and \cite{liu2016dependence}, a malicious adversary without knowledge of how the data are correlated can compromise the practical security of DP.

Encrypting the data with cryptographic operations is another way to protect the data's privacy \cite{gilad2016cryptonets}. 
Ciphers, also called encryption algorithms, are systems for encrypting and decrypting data. 
The primary purpose of encryption methods is to keep sensitive information secret from others by processing readable data into a long series of random or pseudo-random ciphers \cite{carlet2021boolean}. A typical approach is to use secure multi-party computation (SMC) \cite{yao1982protocols}, where each party uses a set of cryptographic methods, and the oblivious transfer scheme to compute a function using their private data jointly \cite{lindell2020secure}. Many IoT devices currently use encryption protocols to protect their dada, e.g., the health care industry \cite{patil2014big} and smart home devices \cite{abu2020towards}. 
In \cite{moriai2019privacy}, \cite{sun2021novel} deep learning is used in combination with encryption to enhance privacy-preserving by keeping high utility gain and maintaining a low leakage rate of sensitive information. 

The database community employs privacy-preserving data mining (PPDM) techniques to ensure that no instances of the database can be attributed to a person in terms of private information \cite{matwin2013privacy,mendes2017privacy,korolova2010privacy}.
In addition to PPDM, many privacy-preserving machine learning (PPML) techniques \cite{psychoula2018deep,zhao2018privacy,sarwate2013signal,rubinstein2009learning,chaudhuri2013near,abadi2016deep} have been proposed to deal with data beyond those in the traditional databases.

In general, previous works of literature ensure that private information cannot be mined and make no assumptions about non-private information.  
On the other hand, our work assumes predefined sets of private and non-private information. Such a formulation makes the proposed data sanitization more effective and provides a flexible tradeoff between privacy and the ability to extract non-private information from the sanitized data.

Researchers have investigated the privacy-enhancing effects of feature selection in several studies \cite{banerjee2011privacy,jafer2015framework}. Data, in this case, is not being released in its entirety but, rather, selected features only. 
By zeroing out feature components in the approximate null space, the work in \cite{xu2017cleaning} proposes a privacy mechanism to minimize confidential information exposure that the client may wish to keep private.
The system described in \cite{enev2012sensorsift} transforms data in such a way that the correlation between data and desired information increases but decreases between data and confidential information.  

Information theorists have studied Privacy-preserving notions under the rubric of information-theoretic privacy \cite{basciftci2016privacy,wang2017privacy,makhdoumi2014information,du2012privacy,tishby2000information}. Information-theoretic privacy has traditionally been quantified by mutual information, which measures how well an adversary can refine its belief about the private features of the data with access to it.
Other works employ mutual information minimization between latent variables in different ways. 
In \cite{creswell2017adversarial, klys2018learning}, a variational autoencoder (VAE) based generative model was proposed using mutual information minimization between the latent space of the VAE and the labeled features. 
In \cite{louizos2015variational}, independence between latent variables is enforced by an additional penalty term based on the  \textit{Maximum Mean Discrepancy}. 
Several other works, such as \cite{edwards2015censoring,hamm2015preserving,ganin2016domain}, have utilized adversarial learning methods in order to learn latent representations that are not directly comparable to ours.

The research community has recently proposed distributed learning architectures that can allow multiple users to share their data to train deep learning models \cite{peteiro2013survey}.
The distribution of information and cooperation between users could lead to the leakage of sensitive information among parties \cite{zhang2018survey}.  
There are concerns about privacy and confidentiality that prevent organizations, such as medical institutions, from fully utilizing distributed deep learning \cite{beaulieu2018privacy,jeon2019privacy}.
To overcome this challenge, researchers have proposed various methods for protecting data privacy in distributed machine learning architectures.
As an alternative to the traditional centralized approach to training artificial intelligence models, federated learning (FL) has recently emerged as a promising alternative. At its core, FL enables multiple parties to train a global model collaboratively without exposing their private data \cite{yang2019federated}.
Numerous practical applications of FL can be found in situations where data is distributed and privacy is essential. For example, it has exhibited exemplary performance and robustness for healthcare systems \cite{xu2020federated} and wireless networks \cite{chen2020joint}. 

\section{Problem Setting}
\subsection{Problem Formulation}

Consider a data owner who wishes to make their data publicly available while preserving the ability to extract useful information and ensuring privacy for sensitive attributes.

Let $X$, $Z$, $U$, and $S$ be random variables representing raw data, released data, utility attributes, and sensitive attributes, respectively. We denote $X$ as continuous high-dimensional raw data, $U$ as utility attributes the user is willing to reveal, $S$ as private attributes the user wishes to hide (e.g., gender, race, and age), and $Z$ as the released (sanitized) data.

We consider instance vectors $\textbf{x}\in\mathbb{R}^{d_{x}}$, $\textbf{z}\in\mathbb{R}^{d_{z}}$, $\textbf{u}\in\mathbb{R}^{d_{u}}$, and $\textbf{s}\in\mathbb{R}^{d_{s}}$, where $d_{z}\ll d_{x}$, corresponding to realizations of $X$, $Z$, $U$, and $S$, respectively. The constraint $d_{z}\ll d_{x}$ ensures that the released data lies in a compressed, lower-dimensional space.

The objective of this work is to develop a stochastic mapping $P(Z|X)$ that takes $X$ as input and outputs $Z$, with the aim of preserving information about the utility variable $U$ while revealing minimal information about the sensitive variable $S$.

As a concrete example, consider facial images: $\textbf{x}^{i}=[x_{1}^{i}, x_{2}^{i}, \ldots, x_{d_{x}}^{i}]^{T}$ represents a face image $i$ with $d_{x}$ pixels; $\textbf{u}^{i}=[u_{1}^{i}, u_{2}^{i}, \ldots, u_{d_{u}}^{i}]^{T}$ represents target utility features (e.g., facial expressions, wearing glasses); $\textbf{s}^{i}=[s_{1}^{i}, s_{2}^{i}, \ldots, s_{d_{s}}^{i}]^{T}$ represents sensitive features (e.g., gender, race, and disability); and $\textbf{z}^{i}=[z_{1}^{i}, z_{2}^{i}, \ldots, z_{d_{z}}^{i}]^{T}$ is the released data designed to remove information about sensitive features $\textbf{s}$ while retaining information about utility features $\textbf{u}$.

\subsection{Bayesian Model}

The raw data $X$ carries inherent utility attributes $U$ and sensitive attributes $S$, with joint distribution $P(X,U,S) = P(U,S)\,P(X \mid U,S)$ in which $U$ and $S$ may themselves be correlated. The encoder defines a stochastic kernel $P(Z \mid X)$, yielding the joint factorization
\begin{equation}
P(X,Z,U,S) = P(U,S)\,P(X \mid U,S)\,P(Z \mid X).
\label{eq:joint}
\end{equation}
For the released representation to be a privacy-meaningful summary, we require that any inference of $U$ or $S$ by an external party rely \emph{only} on $Z$:
\begin{equation}
U \perp X \mid Z, \qquad S \perp X \mid Z.
\label{eq:bottleneck_assumption}
\end{equation}
We treat \eqref{eq:bottleneck_assumption} as a \emph{design objective} enforced by the bottleneck structure of the encoder, not as an inherent property of the data. A neural encoder cannot enforce it exactly; the residual conditional dependence $I(U; X \mid Z)$ is bounded by the slack between our variational utility upper bound and the true conditional entropy $H(U \mid Z)$, and is judged empirically by how close the achieved utility AUC sits to the No-Privacy upper bound. Under \eqref{eq:bottleneck_assumption}, the inference task reduces to learning $P(U \mid Z)$ and $P(S \mid Z)$, illustrated in Figure~\ref{fig:bayesian_network}: GPP shapes $P(Z \mid X)$ so that $P(U \mid Z)$ remains informative (low $H(U \mid Z)$) while $P(S \mid Z)$ collapses toward the prior $P(S)$ (low $I(Z;S)$).

\begin{figure}[!t]
    \centering
    \includegraphics[width=0.8\columnwidth]{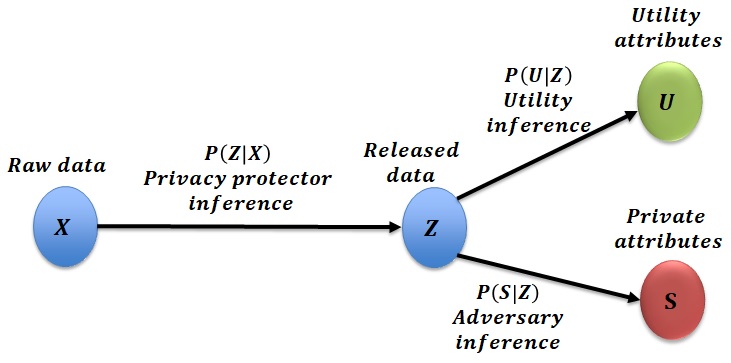}
    \caption{Inference structure of the GPP framework: $X$ is mapped to $Z$ by the privacy protector $P(Z \mid X)$, from which $U$ and $S$ are inferred via $P(U \mid Z)$ and $P(S \mid Z)$.}
    \label{fig:bayesian_network}
\end{figure}

\subsection{Threat Model}
\label{subsec:threat_model}

We adopt an honest-but-curious threat model. The data owner faithfully runs the encoder $T_\theta$ to produce sanitized representations $Z = T_\theta(X)$ and releases them, together with the utility labels $U$, to a downstream consumer. The adversary is any party that observes the released $(Z, U)$ stream and attempts to infer the sensitive attribute $S$. We make the following assumptions explicit:
\begin{itemize}
\item \emph{What the adversary sees.} The adversary observes arbitrarily many sanitized samples $Z^i$ and the corresponding utility labels $U^i$. The adversary does not see $X$ or $S$.
\item \emph{What the adversary knows.} The adversary knows the encoder architecture and may know the encoder parameters $\theta$; privacy must therefore not rely on obscurity of $T_\theta$.
\item \emph{What the adversary does.} The adversary trains an arbitrary post-hoc classifier $\hat{S} = f(Z)$ on a labeled auxiliary dataset and applies it to released $Z$. We model this in Section~\ref{sec:experiments} by re-training a fresh probe classifier on each released representation set.
\item \emph{What the adversary cannot do.} The adversary cannot influence the training of $T_\theta$, query the encoder adaptively, or coerce the data owner. Active and adaptive adversaries are out of scope.
\end{itemize}
The privacy goal is operational: the adversary's classification accuracy on $S$ should be no better than that achievable from the prior $P(S)$ alone. This goal is what the variational objective of Section~\ref{sec:majhead} optimizes for, and what the audit protocol of Section~\ref{sec:experiments} measures. Application scenarios where this threat model fits include health-monitoring systems releasing sensor summaries to clinicians while protecting patient identity, surveillance systems releasing target-specific cues while anonymizing bystanders, and IoT deployments releasing analytics while protecting per-device fingerprints. The federated extension of Section~\ref{sec:distributed} additionally handles the case where the data is partitioned across mutually distrustful clients; we defer that threat model to Section~\ref{sec:distributed}.

\section{THE PROPOSED APPROACH}
\label{sec:majhead}

This section presents the Gaussian Privacy Protector (GPP) framework, establishing rigorous mathematical foundations for privacy-preserving data release. We derive the optimization objective from first principles using information-theoretic analysis and variational inference, then present the complete algorithmic framework.

\subsection{Problem Formulation}

Consider a data owner who possesses high-dimensional data $X$ containing both sensitive attributes $S$ (to be protected) and utility attributes $U$ (to be preserved). The goal is to learn a stochastic mapping that transforms $X$ into sanitized data $Z$ such that:
\begin{enumerate}
    \item \textbf{Privacy:} An adversary cannot reliably infer $S$ from $Z$
    \item \textbf{Utility:} A legitimate user can accurately infer $U$ from $Z$
\end{enumerate}

Formally, we seek an optimal probabilistic mapping $P(Z|X)$ that minimizes information leakage about sensitive attributes while maximizing information preservation about utility attributes:
\begin{equation}
\begin{aligned}
P(Z|X)^{*} &= \argmin_{P(Z|X) \in \mathbb{P}} I(Z;S) \\
&\quad \text{subject to} \quad I(Z;U) \geq \gamma,
\end{aligned}
\label{eq:main_problem}
\end{equation}
where $\gamma > 0$ is the minimum required utility level, $\mathbb{P}$ denotes the set of all valid probabilistic mappings, and $I(\cdot;\cdot)$ represents mutual information.

\subsection{Equivalent Formulation via Conditional Entropy}

Using the fundamental relationship between mutual information and entropy:
\begin{equation}
I(Z;U) = H(U) - H(U|Z),
\label{eq:mi_entropy}
\end{equation}
where $H(\cdot)$ denotes entropy and $H(\cdot|\cdot)$ denotes conditional entropy, the utility constraint in Eq.~\eqref{eq:main_problem} can be rewritten as:
\begin{equation}
H(U) - H(U|Z) \geq \gamma \implies H(U|Z) \leq H(U) - \gamma \triangleq \bar{\gamma}.
\label{eq:utility_constraint}
\end{equation}

Since $H(U)$ is a constant determined by the data distribution, the optimization problem becomes:
\begin{equation}
\begin{aligned}
P(Z|X)^{*} &= \argmin_{P(Z|X) \in \mathbb{P}} I(Z;S) \\
&\quad \text{subject to} \quad H(U|Z) \leq \bar{\gamma}.
\end{aligned}
\label{eq:reformed_problem}
\end{equation}

This formulation reveals the core trade-off: minimizing $I(Z;S)$ removes private information, while constraining $H(U|Z)$ ensures utility information remains recoverable.

\subsection{Lagrangian Relaxation}

To solve the constrained optimization problem in Eq.~\eqref{eq:reformed_problem}, we employ Lagrangian relaxation with multiplier $\beta > 0$:
\begin{equation}
\mathcal{L}(P(Z|X)) = I(Z;S) + \beta \cdot H(U|Z).
\label{eq:lagrangian}
\end{equation}

The optimal mapping minimizes this objective:
\begin{equation}
P(Z|X)^{*} = \argmin_{P(Z|X) \in \mathbb{P}} \left[ I(Z;S) + \beta \cdot H(U|Z) \right].
\label{eq:lagrangian_opt}
\end{equation}

The Lagrange multiplier $\beta$ controls the privacy-utility trade-off:
\begin{itemize}
    \item $\beta \to 0$: Prioritizes privacy (minimizes $I(Z;S)$)
    \item $\beta \to \infty$: Prioritizes utility (minimizes $H(U|Z)$)
\end{itemize}

\subsection{Neural Network Parameterization}

Since searching over all probabilistic mappings $\mathbb{P}$ is intractable, we parameterize the mapping using a neural network encoder $T_\theta: \mathcal{X} \to \mathcal{Z}$ with parameters $\theta \in \Theta$. The optimization becomes:
\begin{equation}
\theta^{*} = \argmin_{\theta \in \Theta} \left[ I(T_\theta(X); S) + \beta \cdot H(U | T_\theta(X)) \right].
\label{eq:neural_objective}
\end{equation}

For notational convenience, we denote $Z = T_\theta(X)$ throughout the remainder of this section.

\subsection{Variational Bounds for Tractable Optimization}

Direct optimization of Eq.~\eqref{eq:neural_objective} is intractable because both $I(Z;S)$ and $H(U|Z)$ require knowledge of the true posterior distributions $P(S|Z)$ and $P(U|Z)$, which are unavailable in closed form. We derive tractable variational bounds for each term.

\subsubsection{Privacy Term: Variational Lower Bound on $I(Z;S)$}

We seek a tractable bound on the mutual information $I(Z;S)$. Starting from the definition of mutual information:
\begin{equation}
I(Z;S) = H(S) - H(S|Z).
\label{eq:mi_definition}
\end{equation}

The conditional entropy can be written as:
\begin{equation}
H(S|Z) = \mathbb{E}_{P(Z,S)}[-\log P(S|Z)] = -\mathbb{E}_{P(Z,S)}[\log P(S|Z)].
\label{eq:cond_entropy_s}
\end{equation}

Substituting Eq.~\eqref{eq:cond_entropy_s} into Eq.~\eqref{eq:mi_definition}:
\begin{equation}
I(Z;S) = H(S) + \mathbb{E}_{P(Z,S)}[\log P(S|Z)].
\label{eq:mi_expanded}
\end{equation}

Since $P(S|Z)$ is intractable, we introduce a variational approximation $Q_\phi(S|Z)$ parameterized by an adversary network with parameters $\phi$.

\begin{theorem}[Variational Lower Bound on Mutual Information]
\label{thm:mi_bound}
For any distribution $Q_\phi(S|Z)$, the mutual information satisfies:
\begin{equation}
I(Z;S) \geq H(S) + \mathbb{E}_{P(Z,S)}[\log Q_\phi(S|Z)],
\label{eq:mi_lower_bound}
\end{equation}
with equality if and only if $Q_\phi(S|Z) = P(S|Z)$ almost everywhere.
\end{theorem}

\begin{proof}
Starting from Eq.~\eqref{eq:mi_expanded}:
\begin{align}
I(Z;S) &= H(S) + \mathbb{E}_{P(Z,S)}[\log P(S|Z)] \label{eq:proof1_step1}\\
&= H(S) + \mathbb{E}_{P(Z,S)}\!\left[\log \tfrac{Q_\phi(S|Z) \cdot P(S|Z)}{Q_\phi(S|Z)}\right] \label{eq:proof1_step2}\\
&= H(S) + \mathbb{E}_{P(Z,S)}[\log Q_\phi(S|Z)] \notag\\
&\quad + \mathbb{E}_{P(Z,S)}\!\left[\log \tfrac{P(S|Z)}{Q_\phi(S|Z)}\right] \label{eq:proof1_step3}\\
&= H(S) + \mathbb{E}_{P(Z,S)}[\log Q_\phi(S|Z)] \notag\\
&\quad + \mathbb{E}_{P(Z)}\!\left[D_{\text{KL}}(P(S|Z) \,\|\, Q_\phi(S|Z))\right]. \label{eq:proof1_step4}
\end{align}
In Eq.~\eqref{eq:proof1_step2}, we multiply and divide by $Q_\phi(S|Z)$. In Eq.~\eqref{eq:proof1_step3}, we use the logarithm property. In Eq.~\eqref{eq:proof1_step4}, we recognize that the last term is the expected KL divergence.

Since the Kullback-Leibler divergence is non-negative, $D_{\text{KL}}(\cdot \| \cdot) \geq 0$, we obtain:
\begin{equation}
I(Z;S) \geq H(S) + \mathbb{E}_{P(Z,S)}[\log Q_\phi(S|Z)].
\label{eq:mi_bound_final}
\end{equation}
The bound is tight when $Q_\phi(S|Z) = P(S|Z)$, which makes the KL divergence zero.
\end{proof}

The tightest lower bound is achieved by maximizing over the variational parameters:
\begin{equation}
I(Z;S) \geq \max_{\phi \in \Phi} \left\{ H(S) + \mathbb{E}_{P(Z,S)}[\log Q_\phi(S|Z)] \right\}.
\label{eq:mi_bound_tight}
\end{equation}

\begin{remark}[Adversary Capacity]
\label{rem:adversary_capacity}
The tightness of the variational bound depends critically on the expressiveness of the adversary network $Q_\phi$. If the adversary has insufficient capacity to approximate the true posterior $P(S|Z)$, the bound becomes loose, potentially giving a false sense of privacy. In practice, we ensure the adversary network has comparable or greater capacity than the encoder network. This design choice guarantees that if the adversary fails to extract sensitive information, it is because such information is genuinely absent from $Z$, not because the adversary is too weak.
\end{remark}

Since $H(S)$ is a constant independent of the optimization variables, we can drop it for optimization purposes:
\begin{equation}
I(Z;S) \geq \max_{\phi \in \Phi} \mathbb{E}_{P(Z,S)}[\log Q_\phi(S|Z)] + \text{const}.
\label{eq:mi_bound_opt}
\end{equation}

\subsubsection{Utility Term: Variational Upper Bound on $H(U|Z)$}

For the conditional entropy term, we derive a variational upper bound.

\begin{theorem}[Variational Upper Bound on Conditional Entropy]
\label{thm:entropy_bound}
For any distribution $Q_\psi(U|Z)$ parameterized by $\psi$, the conditional entropy satisfies:
\begin{equation}
H(U|Z) \leq \mathbb{E}_{P(Z,U)}[-\log Q_\psi(U|Z)],
\label{eq:entropy_upper_bound}
\end{equation}
with equality if and only if $Q_\psi(U|Z) = P(U|Z)$ almost everywhere.
\end{theorem}

\begin{proof}
Starting from the definition of conditional entropy:
\begin{align}
H(U|Z) &= \mathbb{E}_{P(Z,U)}[-\log P(U|Z)] \label{eq:proof2_step1}\\
&= \mathbb{E}_{P(Z,U)}\!\left[-\log Q_\psi(U|Z) - \log \tfrac{P(U|Z)}{Q_\psi(U|Z)}\right] \label{eq:proof2_step2}\\
&= \mathbb{E}_{P(Z,U)}[-\log Q_\psi(U|Z)] \notag\\
&\quad - \mathbb{E}_{P(Z)}\!\left[D_{\text{KL}}(P(U|Z) \,\|\, Q_\psi(U|Z))\right]. \label{eq:proof2_step3}
\end{align}
In Eq.~\eqref{eq:proof2_step2}, we add and subtract $\log Q_\psi(U|Z)$. In Eq.~\eqref{eq:proof2_step3}, we recognize the expected KL divergence.

Since $D_{\text{KL}}(\cdot \| \cdot) \geq 0$, we have:
\begin{equation}
H(U|Z) \leq \mathbb{E}_{P(Z,U)}[-\log Q_\psi(U|Z)].
\label{eq:entropy_bound_final}
\end{equation}
Equality holds when $Q_\psi(U|Z) = P(U|Z)$.
\end{proof}

The tightest upper bound is obtained by minimizing over the variational parameters:
\begin{equation}
H(U|Z) \leq \min_{\psi \in \Psi} \mathbb{E}_{P(Z,U)}[-\log Q_\psi(U|Z)].
\label{eq:entropy_bound_tight}
\end{equation}

\begin{remark}
Note the critical difference in optimization directions: the privacy bound requires \emph{maximization} over $\phi$ (adversarial), while the utility bound requires \emph{minimization} over $\psi$ (cooperative). This asymmetry is fundamental to the adversarial training framework.
\end{remark}

\subsection{Combined Variational Objective}

Substituting the variational bounds from Theorems~\ref{thm:mi_bound} and~\ref{thm:entropy_bound} into the Lagrangian objective Eq.~\eqref{eq:lagrangian}, we define the \emph{variational surrogate}
\begin{equation}
\begin{aligned}
\widetilde{\mathcal{L}}(\theta;\phi,\psi) \triangleq{}& H(S) + \mathbb{E}_{P(Z,S)}[\log Q_\phi(S|Z)] \\
&+ \beta \cdot \mathbb{E}_{P(Z,U)}[-\log Q_\psi(U|Z)],
\end{aligned}
\label{eq:variational_surrogate}
\end{equation}
which, by Theorems~\ref{thm:mi_bound} and~\ref{thm:entropy_bound}, satisfies the saddle-point identity
\begin{equation}
\mathcal{L}(\theta) \;=\; I(Z;S) + \beta \cdot H(U|Z) \;=\; \max_{\phi \in \Phi} \, \min_{\psi \in \Psi} \, \widetilde{\mathcal{L}}(\theta;\phi,\psi),
\label{eq:combined_bound}
\end{equation}
with equality attained at $Q_{\phi^*}(S|Z)=P(S|Z)$ and $Q_{\psi^*}(U|Z)=P(U|Z)$.

To minimize the original objective $\mathcal{L}(\theta)$, we therefore solve the saddle-point problem on the right-hand side of Eq.~\eqref{eq:combined_bound}: the adversary $\phi$ maximizes its inner term to tighten the lower bound on $I(Z;S)$, the utility classifier $\psi$ minimizes its inner term to tighten the upper bound on $H(U|Z)$, and the encoder $\theta$ then minimizes the resulting surrogate. When both inner optimizations are attained, the surrogate equals the true information-theoretic objective, and the encoder's update direction coincides with the gradient of $\mathcal{L}(\theta)$. This is analogous to the generator-discriminator dynamics in Generative Adversarial Networks \cite{goodfellow2014generative}. The complete optimization problem becomes:
\begin{equation}
\boxed{
\theta^{*} = \argmin_{\theta \in \Theta} \left\{ \max_{\phi \in \Phi} \mathbb{E}[\log Q_\phi(S|Z)] + \beta \cdot \min_{\psi \in \Psi} \mathbb{E}[-\log Q_\psi(U|Z)] \right\}
}
\label{eq:main_objective}
\end{equation}

This reveals a \textbf{hybrid minimax-minimization} structure:
\begin{itemize}
    \item \textbf{Privacy (adversarial game):} $\min_\theta \max_\phi$ --- the encoder $T_\theta$ tries to fool the adversary $Q_\phi$
    \item \textbf{Utility (cooperative game):} $\min_\theta \min_\psi$ --- the encoder $T_\theta$ helps the utility classifier $Q_\psi$
\end{itemize}

In practice, directly solving this nested optimization is challenging. We adopt an alternating optimization approach where we iteratively: (1) fix $\theta$ and optimize $\phi$ and $\psi$ to tighten the variational bounds, then (2) fix $\phi$ and $\psi$ and optimize $\theta$ to minimize the surrogate objective. This approach is analogous to training procedures used in Generative Adversarial Networks (GANs) \cite{goodfellow2014generative} and has been shown to be effective for similar min-max problems.

\subsection{Cross-Entropy Formulation}

Recognizing that the expectations in Eq.~\eqref{eq:main_objective} correspond to cross-entropy losses, we can rewrite the objective in a more practical form.

\begin{definition}[Cross-Entropy Loss]
For a true distribution $p$ and a predicted distribution $q$, the cross-entropy is defined as:
\begin{equation}
\text{CE}(p, q) = \mathbb{E}_{p}[-\log q].
\label{eq:ce_definition}
\end{equation}
\end{definition}

We can relate the terms in Eq.~\eqref{eq:main_objective} to cross-entropy as follows:
\begin{itemize}
    \item Utility term: $\mathbb{E}_{P(Z,U)}[-\log Q_\psi(U|Z)] = \text{CE}(U, Q_\psi(Z))$
    \item Privacy term: $\mathbb{E}_{P(Z,S)}[\log Q_\phi(S|Z)] = -\text{CE}(S, Q_\phi(Z))$
\end{itemize}

Substituting these into Eq.~\eqref{eq:main_objective} and using $\max_\phi \mathbb{E}[\log Q_\phi(S|Z)] = -\min_\phi \text{CE}(S, Q_\phi(Z))$:
\begin{equation}
\theta^{*} = \argmin_{\theta} \left\{ -\min_{\phi} \text{CE}(S, Q_\phi(Z)) + \beta \cdot \min_{\psi} \text{CE}(U, Q_\psi(Z)) \right\}.
\label{eq:ce_objective_intermediate}
\end{equation}

Reordering terms for clarity:
\begin{equation}
\theta^{*} = \argmin_{\theta} \left\{ \beta \cdot \min_{\psi} \text{CE}(U, Q_\psi(Z)) - \min_{\phi} \text{CE}(S, Q_\phi(Z)) \right\}.
\label{eq:ce_objective}
\end{equation}
The structure is now transparent: the encoder minimizes a weighted combination of the utility classifier's best-achievable CE (cooperative term, with weight $+\beta$) and the negative of the adversary's best-achievable CE (adversarial term, with weight $-1$). Pushing up the adversary's best-achievable CE is precisely what makes the most powerful adversary fail.

For multiple utility attributes $\{u_1, \ldots, u_{d_u}\}$ and sensitive attributes $\{s_1, \ldots, s_{d_s}\}$, we assume conditional independence given $Z$:
\begin{equation}
P(U|Z) = \prod_{j=1}^{d_u} P(u_j|Z), \quad P(S|Z) = \prod_{j=1}^{d_s} P(s_j|Z).
\label{eq:conditional_independence}
\end{equation}
This factorization assumption allows us to decompose the multi-attribute problem into independent per-attribute classifiers, which simplifies optimization and scales linearly with the number of attributes. When attributes are correlated (e.g., age and wrinkles in facial images), this assumption is approximate, but empirical results demonstrate that the framework remains effective. Under this assumption, the objective becomes:
\begin{equation}
\boxed{
\mathcal{L}_{\text{GPP}}(\theta, \psi, \phi) = \beta \sum_{j=1}^{d_u} \text{CE}(u_j, Q_{\psi_j}(Z)) - \sum_{j=1}^{d_s} \text{CE}(s_j, Q_{\phi_j}(Z))
}
\label{eq:final_objective}
\end{equation}

The optimization proceeds as follows:
\begin{itemize}
    \item $\theta$ is minimized with respect to $\mathcal{L}_{\text{GPP}}$: The encoder learns to protect privacy (increase adversary's cross-entropy) and preserve utility (decrease utility classifier's cross-entropy)
    \item $\psi_j$ is minimized with respect to $\text{CE}(u_j, Q_{\psi_j})$: Utility classifiers learn to extract utility information
    \item $\phi_j$ is minimized with respect to $\text{CE}(s_j, Q_{\phi_j})$: Adversary classifiers learn to extract private information
\end{itemize}

The key insight is that the encoder $T_\theta$ faces an \textbf{adversarial relationship} with the privacy classifiers (their loss is subtracted in $\mathcal{L}_{\text{GPP}}$, so minimizing $\mathcal{L}_{\text{GPP}}$ increases their loss) and a \textbf{cooperative relationship} with the utility classifiers (their loss is added with positive weight $\beta$, so minimizing $\mathcal{L}_{\text{GPP}}$ decreases their loss).

\subsection{Gaussian Latent Space Regularization}

To ensure the latent representation $Z$ has well-behaved distributional properties suitable for privacy protection, we impose a Gaussian prior using the variational autoencoder (VAE) framework \cite{kingma2013auto}.

\subsubsection{Reparameterization Trick}

The encoder network $T_\theta$ outputs parameters of a Gaussian distribution rather than deterministic values:
\begin{equation}
T_\theta(X) \to (\mu_\theta(X), \Sigma_\theta(X)),
\label{eq:encoder_output}
\end{equation}
where $\mu_\theta(X) \in \mathbb{R}^{d_z}$ is the mean vector and $\Sigma_\theta(X) = L_\theta(X) L_\theta(X)^\top$ is the covariance matrix, with $L_\theta$ being the Cholesky factor.

The latent representation is sampled using the reparameterization trick \cite{kingma2013auto}:
\begin{equation}
Z = \mu_\theta(X) + L_\theta(X) \odot \epsilon, \quad \epsilon \sim \mathcal{N}(0, I),
\label{eq:reparam}
\end{equation}
where $\odot$ denotes element-wise multiplication. This reparameterization allows gradients to flow through the sampling operation, enabling end-to-end training via backpropagation.

\subsubsection{KL Divergence Regularization}

To encourage the latent distribution to remain close to a standard Gaussian prior, we add a KL divergence penalty:
\begin{equation}
\mathcal{L}_{\text{KL}}(\theta) = D_{\text{KL}}\left( \mathcal{N}(\mu_\theta(X), \Sigma_\theta(X)) \| \mathcal{N}(0, I) \right).
\label{eq:kl_penalty}
\end{equation}

For diagonal covariance $\Sigma = \text{diag}(\sigma_1^2, \ldots, \sigma_{d_z}^2)$, this has the closed-form expression:
\begin{equation}
\mathcal{L}_{\text{KL}}(\theta) = \frac{1}{2} \sum_{i=1}^{d_z} \left( \mu_i^2 + \sigma_i^2 - \log \sigma_i^2 - 1 \right).
\label{eq:kl_closed_form}
\end{equation}

The Gaussian regularization serves multiple purposes: (1) it prevents the encoder from producing degenerate representations, (2) it provides a smooth latent space that generalizes better to unseen data, and (3) it adds controlled stochasticity that can enhance privacy by making the mapping from $X$ to $Z$ non-deterministic.

\subsection{Complete GPP Objective Function}

Combining the privacy-utility objective (Eq.~\ref{eq:final_objective}) with Gaussian regularization yields the complete GPP loss:
\begin{equation}
\boxed{
\begin{aligned}
\mathcal{L}_{\text{GPP}}(\theta, \psi, \phi) = &\underbrace{\beta \sum_{j=1}^{d_u} \text{CE}(u_j, Q_{\psi_j}(Z))}_{\text{Utility Loss (minimize)}} \\
&- \underbrace{\sum_{j=1}^{d_s} \text{CE}(s_j, Q_{\phi_j}(Z))}_{\text{Privacy Loss (adversarial)}} \\
&+ \underbrace{\lambda \cdot D_{\text{KL}}(\mathcal{N}(\mu_\theta, \Sigma_\theta) \| \mathcal{N}(0, I))}_{\text{Gaussian Regularization}}
\end{aligned}
}
\label{eq:complete_objective}
\end{equation}

where $\lambda > 0$ controls the strength of the Gaussian prior regularization.

\subsection{Optimization Strategy}

The optimization involves three sets of parameters with different objectives, summarized in Table~\ref{tab:param_roles}.

\begin{table}[!t]
\centering
\caption{Parameter Roles in GPP Optimization}
\label{tab:param_roles}
\begin{tabular}{lll}
\hline
\textbf{Parameters} & \textbf{Network} & \textbf{Objective} \\
\hline
$\theta$ & Encoder $T_\theta$ & Minimize $\mathcal{L}_{\text{GPP}}$ \\
$\psi = \{\psi_1, \ldots, \psi_{d_u}\}$ & Utility classifiers & Minimize CE$(u_j, Q_{\psi_j})$ \\
$\phi = \{\phi_1, \ldots, \phi_{d_s}\}$ & Adversary classifiers & Minimize CE$(s_j, Q_{\phi_j})$ \\
\hline
\end{tabular}
\end{table}

\begin{remark}[Adversarial Dynamics]
Although all networks minimize their respective objectives, the adversarial relationship emerges from the structure of $\mathcal{L}_{\text{GPP}}$. When the encoder $\theta$ minimizes $\mathcal{L}_{\text{GPP}}$:
\begin{itemize}
    \item The utility CE terms have positive coefficients ($+\beta$), so minimizing $\mathcal{L}_{\text{GPP}}$ encourages $\theta$ to \emph{reduce} utility classification loss (cooperative).
    \item The privacy CE terms have negative coefficients ($-1$), so minimizing $\mathcal{L}_{\text{GPP}}$ encourages $\theta$ to \emph{increase} adversary classification loss (adversarial).
\end{itemize}
Meanwhile, the adversary networks $\phi_j$ independently minimize their own CE losses, creating the classic min-max game: the encoder tries to make the adversary fail, while the adversary tries to succeed.
\end{remark}

\subsubsection{Gradient Analysis}

For the encoder parameters $\theta$, the gradient of the complete objective is:
\begin{equation}
\begin{aligned}
\nabla_\theta \mathcal{L}_{\text{GPP}} ={}& \beta \sum_{j=1}^{d_u} \nabla_\theta \text{CE}(u_j, Q_{\psi_j}(Z)) \\
&- \sum_{j=1}^{d_s} \nabla_\theta \text{CE}(s_j, Q_{\phi_j}(Z)) \\
&+ \lambda \nabla_\theta \mathcal{L}_{\text{KL}}.
\end{aligned}
\label{eq:gradient}
\end{equation}
The negative sign before the privacy term means:
\begin{itemize}
    \item When the adversary successfully classifies $s_j$ (low CE), the gradient pushes $\theta$ to \emph{increase} the adversary's loss
    \item This creates the adversarial dynamic where the encoder learns to confuse the adversary
\end{itemize}

\subsection{Training Algorithm}

Algorithm~\ref{alg:gpp_training} presents the complete GPP training procedure with alternating optimization.

\begin{algorithm}[t]
\caption{Gaussian Privacy Protector (GPP) Training}
\label{alg:gpp_training}
\begin{algorithmic}[1]
\Require Dataset $\mathcal{D} = \{(\mathbf{x}^i, \mathbf{u}^i, \mathbf{s}^i)\}_{i=1}^{n}$; batch size $b$; adversary/utility update steps $k$; trade-off parameter $\beta$; regularization weight $\lambda$; learning rate $\alpha$
\Ensure Trained encoder $T_\theta$
\State Initialize $T_\theta$, $\{Q_{\psi_j}\}_{j=1}^{d_u}$, $\{Q_{\phi_j}\}_{j=1}^{d_s}$ with random weights
\While{not converged}
    \State \textcolor{blue}{\texttt{// Phase 1: Train utility and adversary classifiers}}
    \For{$t = 1$ to $k$}
        \State Sample mini-batch $\mathcal{B} = \{(\mathbf{x}^i, \mathbf{u}^i, \mathbf{s}^i)\}_{i=1}^{b}$ from $\mathcal{D}$
        \State $\boldsymbol{\mu}, \boldsymbol{L} \gets T_\theta(\{\mathbf{x}^i\})$ \Comment{Encoder forward pass}
        \State Sample $\boldsymbol{\epsilon} \sim \mathcal{N}(0, I)$
        \State $\{\mathbf{z}^i\} \gets \boldsymbol{\mu} + \boldsymbol{L} \odot \boldsymbol{\epsilon}$ \Comment{Reparameterization}
        \For{$j = 1$ to $d_u$} \Comment{Update utility classifiers}
            \State $\mathcal{L}_{\psi_j} \gets \frac{1}{b} \sum_{i=1}^{b} \text{CE}(u_j^i, Q_{\psi_j}(\mathbf{z}^i))$
            \State $\psi_j \gets \psi_j - \alpha \cdot \text{Adam}(\nabla_{\psi_j} \mathcal{L}_{\psi_j})$
        \EndFor
        \For{$j = 1$ to $d_s$} \Comment{Update adversary classifiers}
            \State $\mathcal{L}_{\phi_j} \gets \frac{1}{b} \sum_{i=1}^{b} \text{CE}(s_j^i, Q_{\phi_j}(\mathbf{z}^i))$
            \State $\phi_j \gets \phi_j - \alpha \cdot \text{Adam}(\nabla_{\phi_j} \mathcal{L}_{\phi_j})$
        \EndFor
    \EndFor
    \State \textcolor{blue}{\texttt{// Phase 2: Train encoder}}
    \State Sample mini-batch $\mathcal{B} = \{(\mathbf{x}^i, \mathbf{u}^i, \mathbf{s}^i)\}_{i=1}^{b}$ from $\mathcal{D}$
    \State $\boldsymbol{\mu}, \boldsymbol{L} \gets T_\theta(\{\mathbf{x}^i\})$
    \State Sample $\boldsymbol{\epsilon} \sim \mathcal{N}(0, I)$
    \State $\{\mathbf{z}^i\} \gets \boldsymbol{\mu} + \boldsymbol{L} \odot \boldsymbol{\epsilon}$
    \State Compute KL divergence: $\mathcal{L}_{\text{KL}} \gets \frac{1}{2} \sum_{l=1}^{d_z} \left( \mu_l^2 + \sigma_l^2 - \log \sigma_l^2 - 1 \right)$
    \State Compute encoder loss:
    \begin{equation*}
    \mathcal{L}_\theta \gets \frac{1}{b} \sum_{i=1}^{b} \left[ \beta \sum_{j=1}^{d_u} \text{CE}(u_j^i, Q_{\psi_j}(\mathbf{z}^i)) - \sum_{j=1}^{d_s} \text{CE}(s_j^i, Q_{\phi_j}(\mathbf{z}^i)) \right] + \lambda \cdot \mathcal{L}_{\text{KL}}
    \end{equation*}
    \State $\theta \gets \theta - \alpha \cdot \text{Adam}(\nabla_\theta \mathcal{L}_\theta)$
\EndWhile
\State \Return $T_\theta$
\end{algorithmic}
\end{algorithm}

\subsection{Theoretical Analysis}

We now analyze the theoretical properties of the GPP framework.

\subsubsection{Privacy-Utility Trade-off}

The fundamental trade-off between privacy and utility in the GPP framework can be understood through the lens of rate-distortion theory and the information bottleneck principle \cite{tishby2000information}.

\begin{remark}[Privacy-Utility Trade-off]
\label{rem:tradeoff}
Let $\theta^*$ be the optimal encoder parameters obtained by minimizing Eq.~\eqref{eq:complete_objective}, and let $Z^* = T_{\theta^*}(X)$. The achievable privacy-utility pairs $(I(Z^*;S), I(Z^*;U))$ lie on a trade-off curve parameterized by $\beta$:
\begin{itemize}
    \item As $\beta \to 0$: $I(Z^*;S) \to 0$ (maximum privacy), but $I(Z^*;U)$ may decrease
    \item As $\beta \to \infty$: $I(Z^*;U) \to I(X;U)$ (maximum utility), but $I(Z^*;S)$ may increase
\end{itemize}
The bottleneck dimension $d_z$ further constrains this trade-off: smaller $d_z$ limits the total information $Z$ can carry about $X$, forcing a more stringent trade-off.
\end{remark}

\subsubsection{Adversary Performance Bound}

We can bound the adversary's performance using Fano's inequality \cite{cover2012elements}.

\begin{proposition}[Adversary Accuracy Bound]
\label{prop:adversary_bound}
Let $S$ be a discrete sensitive attribute with $|\mathcal{S}|$ classes, and assume $S$ is uniformly distributed so that $H(S) = \log |\mathcal{S}|$. For any estimator $\hat{S}(Z)$ of $S$ based on $Z$, the probability of error is lower bounded by:
\begin{equation}
P(\hat{S}(Z) \neq S) \geq \frac{H(S|Z) - 1}{\log |\mathcal{S}|} = \frac{H(S) - I(Z;S) - 1}{\log |\mathcal{S}|}.
\label{eq:fano_bound}
\end{equation}
Consequently, if GPP achieves $I(Z^*;S) \approx 0$, then the error probability satisfies:
\begin{equation}
P(\hat{S}(Z^*) \neq S) \geq \frac{\log |\mathcal{S}| - 1}{\log |\mathcal{S}|} = 1 - \frac{1}{\log |\mathcal{S}|}
\label{eq:adversary_bound}
\end{equation}
which approaches $1$ (i.e., the adversary fails almost surely) for large $|\mathcal{S}|$.
\end{proposition}

\begin{proof}
By Fano's inequality \cite{cover2012elements}, for any estimator $\hat{S}$ of $S$ based on $Z$:
\begin{equation}
H(S|Z) \leq H_b(P_e) + P_e \log(|\mathcal{S}| - 1),
\label{eq:fano_ineq}
\end{equation}
where $P_e = P(\hat{S}(Z) \neq S)$ is the error probability and $H_b(\cdot)$ is the binary entropy function. Using $H_b(P_e) \leq 1$ and $\log(|\mathcal{S}| - 1) \leq \log |\mathcal{S}|$, we obtain:
\begin{equation}
H(S|Z) \leq 1 + P_e \log |\mathcal{S}|.
\end{equation}
Rearranging and substituting $H(S|Z) = H(S) - I(Z;S)$ yields:
\begin{equation}
P_e \geq \frac{H(S) - I(Z;S) - 1}{\log |\mathcal{S}|}.
\end{equation}
When $I(Z;S) \to 0$, we have $H(S|Z) \to H(S)$, meaning $Z$ provides negligible information about $S$, and any estimator's error probability approaches $(H(S) - 1)/\log|\mathcal{S}|$.
\end{proof}

\begin{remark}[Informativeness regime of Proposition~\ref{prop:adversary_bound}]
\label{rem:fano_regime}
The bound in Eq.~\eqref{eq:adversary_bound} is non-trivial only when $H(S) > 1$ (in nats), i.e., when the sensitive attribute has many classes. For binary $S$ (e.g., the even/odd label on MNIST or the gender label on CelebA), $H(S) \le \log 2 < 1$, and the right-hand side of Eq.~\eqref{eq:fano_bound} is non-positive, rendering the bound vacuous. The bound is therefore most informative for the HAPT-Recognition setting ($|\mathcal{S}| = 30$ subject identities), where it certifies $P_e \gtrsim 1 - 1/\log|\mathcal{S}| \approx 0.71$ as $I(Z^*;S) \to 0$. For binary sensitive attributes, the empirical adversary AUC approaching $0.5$ in our experiments serves as the operational privacy guarantee, complementing rather than relying on Fano's bound.
\end{remark}

\subsection{Connection to Related Frameworks}

The GPP framework unifies and extends several existing information-theoretic approaches:

\begin{itemize}
    \item \textbf{Information Bottleneck} \cite{tishby2000information}: Setting $S = X$ reduces GPP to the standard information bottleneck, which compresses $X$ while preserving information about $U$. The objective becomes $\min I(Z;X) - \beta I(Z;U)$.
    
    \item \textbf{Privacy Funnel} \cite{makhdoumi2014information}: Setting $U = X$ reduces GPP to the privacy funnel, which reveals information about $X$ while hiding $S$. The objective becomes $\min I(Z;S) - \beta I(Z;X)$.
    
    \item \textbf{Adversarial Representation Learning} \cite{edwards2015censoring}: The adversarial component of GPP mirrors domain-adversarial training, but GPP provides explicit information-theoretic objectives and variational bounds.
    
    \item \textbf{Variational Fair Autoencoder} \cite{louizos2015variational}: Similar to GPP, this approach uses variational inference for fair representations, but GPP additionally incorporates the utility preservation objective and provides a complete distributed learning framework.
\end{itemize}

The key novelties of GPP compared to these frameworks are:
GPP's relationship to these frameworks is one of integration and extension rather than wholesale replacement. The variational bounds we use are essentially those of \cite{edwards2015censoring} and \cite{louizos2015variational}; the Gaussian latent regularizer is the standard VAE construction; the saddle-point training procedure mirrors the adversarial dynamics of \cite{goodfellow2014generative}. What is new in this paper is the distributed extension presented in Section~\ref{sec:distributed}: specifically, (i) the architectural split in which adversary classifiers are kept private to each client while only utility classifiers are aggregated, (ii) the substitution of sanitized-representation transmission for gradient transmission as a response to the gradient-leakage attack vector \cite{lyu2020threats}, and (iii) the compositional privacy bound of Proposition~\ref{prop:distributed_privacy} that this architecture admits under standard IID and honest-aggregator assumptions. The centralized GPP framework (Algorithm~\ref{alg:gpp_training}) is presented primarily as the per-client building block of the distributed framework rather than as an independent contribution.

\section{System Model for Distributed Datasets}
\label{sec:distributed}

This section extends the GPP framework to distributed settings, addressing privacy concerns in federated learning (FL) environments where multiple parties collaboratively train models without sharing raw data.

\subsection{Motivation}

Federated learning has emerged as an effective paradigm for collaborative model training across distributed IoT devices, with a central server coordinating the learning process \cite{cheng2021fine}. The standard FL process consists of four main steps \cite{yang2019federated,alazab2021federated}:

\begin{enumerate}
    \item \textbf{Client Selection:} The aggregator (server) selects participating clients either randomly or using a selection algorithm based on criteria such as data quality or computational resources.
    
    \item \textbf{Parameter Distribution:} The server distributes the current global model parameters to all selected clients.
    
    \item \textbf{Local Training:} Each client trains the model locally using their private data, computing parameter updates (gradients).
    
    \item \textbf{Model Aggregation:} Clients send their updated parameters to the central server, which aggregates them (e.g., via averaging) to update the global model. This process repeats iteratively until convergence.
\end{enumerate}

Despite FL's benefits, significant privacy concerns remain. Many individuals and organizations are hesitant to participate in distributed learning environments, particularly in sensitive domains such as healthcare \cite{xu2021federated}, finance \cite{li2020preserving}, and wireless communications \cite{niknam2020federated}. Even though raw data remains local, several privacy risks persist:

\begin{itemize}
    \item \textbf{Gradient Leakage:} Recent work has demonstrated that raw training data can be reconstructed from shared gradients \cite{lyu2020threats}.
    \item \textbf{Model Inversion Attacks:} Adversaries can infer sensitive attributes from model parameters \cite{fredrikson2015model}.
    \item \textbf{Membership Inference:} Attackers can determine whether specific data points were used in training \cite{shokri2017membership}.
    \item \textbf{Untrusted Aggregators:} The central server may be compromised or malicious \cite{ma2020safeguarding, hou2021mitigating}.
\end{itemize}

Figure~\ref{fig:fl_threat} illustrates the threat model in a typical FL system, where adversaries may exist at the central server, communication channels, or as malicious participants.

To address these challenges, we propose \textbf{Distributed GPP}, which enhances privacy in FL systems by having each client sanitize their data locally using a GPP encoder before sharing. Instead of sharing raw data or gradients computed on raw data, clients share only sanitized representations that preserve utility information while removing sensitive attributes.

\begin{figure}[!t]
    \centering
    \includegraphics[width=0.8\columnwidth]{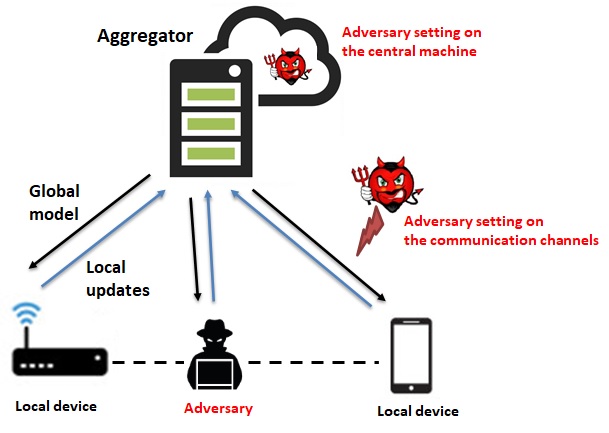}
    \caption{Threat model in federated learning systems. Adversaries may attempt to extract sensitive information at multiple points: the central aggregator, communication channels, or through malicious participants. Distributed GPP addresses these threats by sanitizing data locally before any information leaves the client.}
    \label{fig:fl_threat}
\end{figure}

\subsection{Problem Statement}

Consider a distributed system with $t$ clients (GPP nodes), where each client $m \in \{1, \ldots, t\}$ possesses a local dataset:
\begin{equation}
\mathcal{D}^{m} = \{(\mathbf{x}^{i}_{m}, \mathbf{u}^{i}_{m}, \mathbf{s}^{i}_{m})\}_{i=1}^{n_m},
\label{eq:local_dataset}
\end{equation}
where $n_m$ is the number of samples at client $m$, $\mathbf{x}^{i}_{m} \in \mathbb{R}^{d_x}$ is the raw data, $\mathbf{u}^{i}_{m} \in \mathbb{R}^{d_u}$ represents utility attributes, and $\mathbf{s}^{i}_{m} \in \mathbb{R}^{d_s}$ represents sensitive attributes.

The complete distributed dataset is:
\begin{equation}
\mathcal{D} = \bigcup_{m=1}^{t} \mathcal{D}^{m}, \quad \text{with total size } N = \sum_{m=1}^{t} n_m.
\label{eq:full_dataset}
\end{equation}

\subsubsection{Data Partitioning}

Distributed data can be partitioned in two ways:

\begin{itemize}
    \item \textbf{Horizontal Partitioning:} Each client owns the same set of attributes for different record sets. Formally, all clients share the same feature space $\mathcal{X}$, attribute spaces $\mathcal{U}$ and $\mathcal{S}$, but have disjoint sample sets: $\mathcal{D}^{m} \cap \mathcal{D}^{m'} = \emptyset$ for $m \neq m'$.
    
    \item \textbf{Vertical Partitioning:} Each client owns different attributes for the same set of records. Clients share the same sample identifiers but have disjoint feature spaces.
\end{itemize}

In this work, we focus on \textbf{horizontal partitioning}, which is the most common scenario in federated learning applications such as mobile devices, hospitals, or IoT sensors, where each client has complete feature vectors for their local samples.

\subsubsection{Distributed Objective}

The goal of Distributed GPP is to jointly learn:
\begin{enumerate}
    \item Local encoders $\{T_{\theta_m}\}_{m=1}^{t}$ at each client that sanitize raw data
    \item Shared utility classifiers $\{Q_{\psi_j}\}_{j=1}^{d_u}$ at the aggregator
    \item Adversary classifiers $\{Q_{\phi_j}\}_{j=1}^{d_s}$ (for training purposes)
\end{enumerate}

The distributed optimization objective extends Eq.~\eqref{eq:complete_objective}:
\begin{equation}
\boxed{
\begin{aligned}
\min_{\{\theta_m\}, \{\psi_j\}} \max_{\{\phi_j\}} \quad & \sum_{m=1}^{t} \frac{n_m}{N} \Bigg[ \beta_m \sum_{j=1}^{d_u} \text{CE}(u_j^{(m)}, Q_{\psi_j}(Z^{(m)})) \\
& - \sum_{j=1}^{d_s} \text{CE}(s_j^{(m)}, Q_{\phi_j}(Z^{(m)})) + \lambda_m \mathcal{L}_{\text{KL}}^{(m)} \Bigg]
\end{aligned}
}
\label{eq:distributed_objective}
\end{equation}
where $Z^{(m)} = T_{\theta_m}(X^{(m)})$ is the sanitized representation from client $m$, and $\beta_m$, $\lambda_m$ are client-specific trade-off parameters.

\subsection{Distributed Architecture}

Figure~\ref{fig:distributed_gpp} illustrates the Distributed GPP architecture, which consists of three types of components:

\begin{enumerate}
    \item \textbf{Local GPP Encoders} ($t$ clients): Each client $m$ maintains a local encoder $T_{\theta_m}$ that transforms raw data $\mathbf{x}^{i}_{m}$ into sanitized representations $\mathbf{z}^{i}_{m}$. Sensitive labels $\mathbf{s}^{i}_{m}$ never leave the client.
    
    \item \textbf{Central Aggregator}: Hosts the shared utility classifiers $\{Q_{\psi_j}\}_{j=1}^{d_u}$ that are trained on sanitized data from all clients to perform utility tasks.
    
    \item \textbf{Adversary Classifiers}: $\{Q_{\phi_j}\}_{j=1}^{d_s}$ are used during training to ensure the encoders successfully remove sensitive information. These can be located at the aggregator (for centralized adversarial training) or distributed.
\end{enumerate}

\begin{figure}[!htbp]
    \centering 
    \includegraphics[width=\columnwidth]{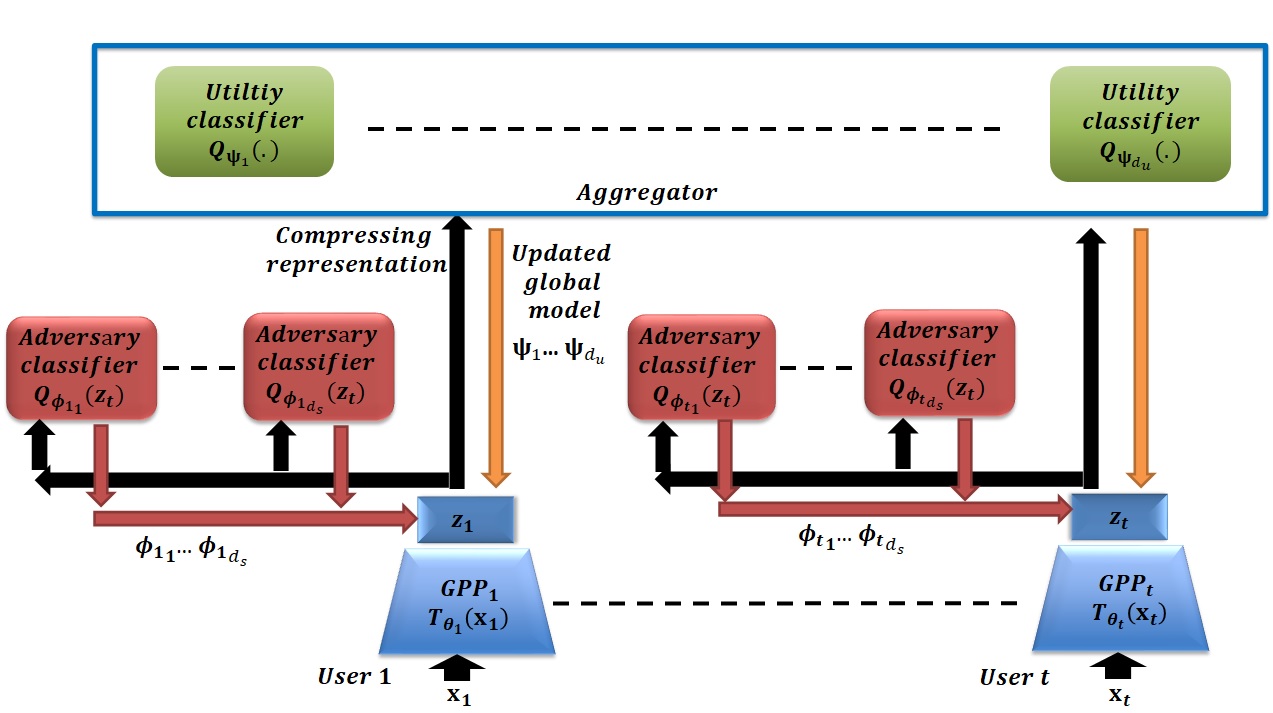}
    \caption{Architecture of Distributed GPP. Each client $m$ has a local encoder $T_{\theta_m}$ that sanitizes raw data before transmission. The aggregator receives only sanitized representations $\mathbf{z}^{i}_{m}$ and utility labels $\mathbf{u}^{i}_{m}$, while sensitive labels $\mathbf{s}^{i}_{m}$ remain local. Utility classifiers are trained centrally on aggregated sanitized data, while adversary classifiers ensure privacy during the training phase.}
    \label{fig:distributed_gpp}
\end{figure}

\subsection{Privacy Guarantees in Distributed GPP}

The Distributed GPP framework provides privacy protection through multiple mechanisms:

\begin{enumerate}
    \item \textbf{Data Sanitization:} Raw data $\mathbf{x}^{i}_{m}$ never leaves the client. Only sanitized representations $\mathbf{z}^{i}_{m}$ are transmitted, which are trained to contain minimal information about sensitive attributes.
    
    \item \textbf{Local Sensitive Labels:} Sensitive labels $\mathbf{s}^{i}_{m}$ are used only locally for adversarial training and are never shared with the aggregator or other clients.
    
    \item \textbf{Dimensionality Reduction:} Since $d_z \ll d_x$, the sanitized representations have significantly lower dimensionality than raw data, limiting the information that can potentially leak.
    
    \item \textbf{Stochastic Encoding:} The Gaussian latent space ensures that the mapping from $\mathbf{x}$ to $\mathbf{z}$ is stochastic, providing additional privacy through randomization.
\end{enumerate}

\begin{proposition}[Privacy Preservation in Distributed GPP, IID and honest aggregator]
\label{prop:distributed_privacy}
Assume (A1)~the client datasets are mutually independent: $(X^{(m)}, S^{(m)}) \perp (X^{(m')}, S^{(m')})$ for $m \neq m'$; (A2)~the aggregator is honest-but-curious, i.e., it follows the protocol but may attempt to infer sensitive attributes from received messages; and (A3)~the only quantities transmitted by client $m$ are the sanitized representations $Z^{(m)} = T_{\theta_m}(X^{(m)})$ and the utility labels $U^{(m)}$. If each local encoder achieves $I(Z^{(m)}; S^{(m)}) \leq \epsilon$, then for every $m' \in \{1, \ldots, t\}$,
\begin{equation}
\begin{aligned}
I\!\left(\{Z^{(m)}, U^{(m)}\}_{m=1}^{t}; S^{(m')}\right) &\le I(Z^{(m')}; S^{(m')}) \\
&\quad + I(U^{(m')}; S^{(m')} \mid Z^{(m')}) \\
&\le \epsilon + \delta,
\end{aligned}
\label{eq:distributed_privacy_bound}
\end{equation}
where $\delta = I(U^{(m')}; S^{(m')} \mid Z^{(m')})$ is the residual leakage carried by the utility labels themselves.
\end{proposition}

\begin{proof}
Under (A1), $(X^{(m)}, U^{(m)}, Z^{(m)}) \perp S^{(m')}$ for $m \neq m'$, so each cross-client term contributes zero mutual information. The chain rule for mutual information then gives
$I(\{Z^{(m)}, U^{(m)}\}_{m}; S^{(m')}) = I(Z^{(m')}, U^{(m')}; S^{(m')})$.
Expanding $I(Z^{(m')}, U^{(m')}; S^{(m')}) = I(Z^{(m')}; S^{(m')}) + I(U^{(m')}; S^{(m')} \mid Z^{(m')})$ and using $I(Z^{(m')}; S^{(m')}) \le \epsilon$ yields the bound.
\end{proof}

\begin{remark}[Limits of Proposition~\ref{prop:distributed_privacy}]
\label{rem:distributed_privacy_limits}
The bound holds under three assumptions worth stating explicitly. (A1) excludes settings where clients share latent factors that correlate $S$ across the boundary (e.g., sites that draw from the same patient demographics). (A2) excludes a malicious aggregator that uses the gradients of the shared utility classifiers $Q_{\psi_j}$ to probe individual clients; mitigating that threat requires additional mechanisms (secure aggregation, differential privacy on the utility classifier updates) which we do not analyze here. (A3) is enforced by Algorithm~\ref{alg:distributed_gpp}, in which only $(\mathbf{z}^i_m, \mathbf{u}^i_m)$ leaves the client. The residual term $\delta$ accounts for the fact that the utility label $U^{(m')}$ may itself carry information about $S^{(m')}$ when the two are correlated; on benchmarks where they are independent (e.g., MNIST parity vs.\ digit-sum), $\delta = 0$.
\end{remark}

\subsection{Distributed Learning Algorithm}

Algorithm~\ref{alg:distributed_gpp} presents the complete Distributed GPP training procedure. The algorithm alternates between two phases:

\begin{enumerate}
    \item \textbf{Local Phase:} Each client updates its local encoder using local data and the current shared classifiers.
    \item \textbf{Aggregation Phase:} The aggregator collects sanitized representations and updates the shared utility and adversary classifiers.
\end{enumerate}

\begin{algorithm*}[!tbp]
\caption{Distributed GPP Training}
\label{alg:distributed_gpp}
\begin{algorithmic}[1]
 
\Require Number of clients $t$; batch size $b$; classifier update steps $k$; trade-off parameters $\{\beta_m\}_{m=1}^{t}$; regularization weights $\{\lambda_m\}_{m=1}^{t}$; learning rate $\alpha$; number of communication rounds $R$
\Ensure Trained local encoders $\{T_{\theta_m}\}_{m=1}^{t}$, utility classifiers $\{Q_{\psi_j}\}_{j=1}^{d_u}$

\State \textbf{Initialize:} Local encoders $\{T_{\theta_m}\}_{m=1}^{t}$, utility classifiers $\{Q_{\psi_j}\}_{j=1}^{d_u}$, adversary classifiers $\{Q_{\phi_j}\}_{j=1}^{d_s}$ with random weights

\For{round $r = 1$ to $R$}
    \State \textcolor{blue}{\texttt{// Phase 1: Local encoding at each client}}
    \For{each client $m = 1$ to $t$ \textbf{in parallel}}
        \State Sample mini-batch $\mathcal{B}_m = \{(\mathbf{x}^i_m, \mathbf{u}^i_m, \mathbf{s}^i_m)\}_{i=1}^{b}$ from $\mathcal{D}^m$
        \State $\boldsymbol{\mu}_m, \boldsymbol{L}_m \gets T_{\theta_m}(\{\mathbf{x}^i_m\})$
        \State Sample $\boldsymbol{\epsilon} \sim \mathcal{N}(0, I)$
        \State $\{\mathbf{z}^i_m\} \gets \boldsymbol{\mu}_m + \boldsymbol{L}_m \odot \boldsymbol{\epsilon}$
        \State Send $\{(\mathbf{z}^i_m, \mathbf{u}^i_m)\}_{i=1}^{b}$ to aggregator \Comment{$\mathbf{s}^i_m$ stays local}
    \EndFor
    
    \State \textcolor{blue}{\texttt{// Phase 2: Aggregator updates classifiers}}
    \State Collect sanitized data: $\mathcal{Z} \gets \bigcup_{m=1}^{t} \{(\mathbf{z}^i_m, \mathbf{u}^i_m)\}$
    \For{$\tau = 1$ to $k$}
        \For{$j = 1$ to $d_u$} \Comment{Update utility classifiers}
            \State $\mathcal{L}_{\psi_j} \gets \frac{1}{|\mathcal{Z}|} \sum_{(\mathbf{z}, \mathbf{u}) \in \mathcal{Z}} \text{CE}(u_j, Q_{\psi_j}(\mathbf{z}))$
            \State $\psi_j \gets \psi_j - \alpha \cdot \text{Adam}(\nabla_{\psi_j} \mathcal{L}_{\psi_j})$
        \EndFor
    \EndFor
    \State Broadcast updated $\{Q_{\psi_j}\}_{j=1}^{d_u}$ to all clients
    
    \State \textcolor{blue}{\texttt{// Phase 3: Local encoder and adversary updates}}
    \For{each client $m = 1$ to $t$ \textbf{in parallel}}
        \State Receive updated utility classifiers $\{Q_{\psi_j}\}_{j=1}^{d_u}$
        \For{$\tau = 1$ to $k$}
            \State Sample mini-batch $\mathcal{B}_m = \{(\mathbf{x}^i_m, \mathbf{u}^i_m, \mathbf{s}^i_m)\}_{i=1}^{b}$
            \State $\boldsymbol{\mu}_m, \boldsymbol{L}_m \gets T_{\theta_m}(\{\mathbf{x}^i_m\})$
            \State $\{\mathbf{z}^i_m\} \gets \boldsymbol{\mu}_m + \boldsymbol{L}_m \odot \boldsymbol{\epsilon}$, where $\boldsymbol{\epsilon} \sim \mathcal{N}(0, I)$
            \For{$j = 1$ to $d_s$} \Comment{Update local adversary}
                \State $\mathcal{L}_{\phi_j}^{(m)} \gets \frac{1}{b} \sum_{i=1}^{b} \text{CE}(s_j^i, Q_{\phi_j}^{(m)}(\mathbf{z}^i_m))$
                \State $\phi_j^{(m)} \gets \phi_j^{(m)} - \alpha \cdot \text{Adam}(\nabla_{\phi_j^{(m)}} \mathcal{L}_{\phi_j}^{(m)})$
            \EndFor
        \EndFor
        \State \textcolor{blue}{\texttt{// Update local encoder}}
        \State Compute KL: $\mathcal{L}_{\text{KL}}^{(m)} \gets \frac{1}{2} \sum_{l=1}^{d_z} (\mu_{m,l}^2 + \sigma_{m,l}^2 - \log \sigma_{m,l}^2 - 1)$
        \State Compute encoder loss:
        \begin{equation*}
\begin{aligned}
\mathcal{L}_{\theta_m} \gets \frac{1}{b} \sum_{i=1}^{b} \Bigg[ &\beta_m \sum_{j=1}^{d_u} \text{CE}(u_j^i, Q_{\psi_j}(\mathbf{z}^i_m)) \\
&- \sum_{j=1}^{d_s} \text{CE}(s_j^i, Q_{\phi_j}^{(m)}(\mathbf{z}^i_m)) \Bigg] + \lambda_m \mathcal{L}_{\text{KL}}^{(m)}
\end{aligned}
\end{equation*}
        \State $\theta_m \gets \theta_m - \alpha \cdot \text{Adam}(\nabla_{\theta_m} \mathcal{L}_{\theta_m})$
    \EndFor
\EndFor

\State \Return $\{T_{\theta_m}\}_{m=1}^{t}$, $\{Q_{\psi_j}\}_{j=1}^{d_u}$
\end{algorithmic}
\end{algorithm*}

\subsection{Communication Efficiency}

An important advantage of Distributed GPP is its communication efficiency compared to standard federated learning:

\begin{itemize}
    \item \textbf{Reduced Data Dimensionality:} Clients transmit sanitized representations $\mathbf{z} \in \mathbb{R}^{d_z}$ instead of raw data $\mathbf{x} \in \mathbb{R}^{d_x}$, with $d_z \ll d_x$. For image data where $d_x$ may be tens of thousands of pixels, $d_z$ can be as small as 100-600 dimensions, representing a compression ratio of 10-100x.
    
    \item \textbf{No Gradient Transmission:} Unlike standard FL which requires transmitting model gradients (which can be large for deep networks), Distributed GPP transmits only encoded representations.
    
    \item \textbf{Bandwidth Savings:} The communication cost per sample per round is $O(d_z + d_u)$ for Distributed GPP, compared to $O(|\theta|)$ for gradient-based FL, where $|\theta|$ is the number of model parameters.
\end{itemize}

\begin{remark}[Position relative to standard FL]
We use ``federated'' here in the operational sense that raw data and sensitive labels remain at the client; we do not employ model averaging in the FedAvg sense. Distributed GPP is structurally closer to \emph{split learning with a client-side privacy filter}: each client owns a private encoder $T_{\theta_m}$ and a private adversary $Q_{\phi_j}^{(m)}$, and the aggregator owns the shared utility classifiers $\{Q_{\psi_j}\}$, which are trained directly on the union of sanitized representations rather than on aggregated gradients. Compared to gradient-based FL this design has two consequences. On the privacy side, even the model updates in standard FL have been shown to leak training data \cite{lyu2020threats}; sanitized representations bypass that leakage channel. On the cost side, the per-round payload from a single client is $b(d_z + d_u)$ scalars for Distributed GPP versus $|\theta|$ scalars for gradient-based FL, where $b$ is the per-client batch size. The crossover is therefore $b(d_z + d_u) < |\theta|$: Distributed GPP transmits less than gradient-based FL whenever the local batch is small relative to the encoder size. For the CelebA experiment with $b=64$, $d_z = 256$, $d_u = 1$ and an encoder with $|\theta| \approx 5\times 10^5$ parameters, this yields roughly a $30\times$ reduction in per-round bytes from each client.
\end{remark}

\subsection{Convergence Considerations}

The convergence of Distributed GPP depends on several factors:

\begin{enumerate}
    \item \textbf{Data Heterogeneity:} When client data distributions differ significantly (non-IID data), convergence may be slower. The client-specific parameters $\beta_m$ and $\lambda_m$ can be tuned to account for heterogeneity.
    
    \item \textbf{Adversary Strength:} Local adversary classifiers must be sufficiently powerful to provide meaningful privacy gradients. Weak adversaries may lead to encoders that leak sensitive information.
    
    \item \textbf{Communication Frequency:} More frequent aggregation rounds generally improve convergence but increase communication overhead.
\end{enumerate}

Under standard assumptions of Lipschitz-continuous gradients and bounded variance, the alternating optimization in Algorithm~\ref{alg:distributed_gpp} converges to a stationary point of the distributed objective \eqref{eq:distributed_objective}. A detailed convergence analysis is beyond the scope of this paper but follows similar arguments to those in federated optimization literature \cite{li2020federated}.

\section{EXPERIMENTS}
\label{sec:experiments}

This section presents comprehensive experimental evaluation of the GPP framework. We first describe the experimental setup, including datasets, evaluation metrics, and baseline methods. We then present six key experiments that validate the effectiveness of GPP for privacy-preserving data release.

\subsection{Experimental Setup}

\subsubsection{Datasets}

We evaluate GPP on three benchmark datasets commonly used in privacy-preserving machine learning research:

\begin{itemize}
    \item \textbf{MNIST} \cite{lecun1998mnist}: Hand-written digit images ($28 \times 28$ pixels, $d_x = 784$). Following \cite{hamm2016learning}, we construct a two-digit composite dataset where the utility attribute $U$ is the sum of two digits (0-18), and the sensitive attribute $S$ is whether the sum is even or odd (binary classification).
    
    \item \textbf{CelebA} \cite{liu2015faceattributes}: Celebrity face images ($64 \times 64 \times 3$ RGB, $d_x = 12,288$). We use ``Smiling'' as the utility attribute and ``Gender'' as the sensitive attribute, following the privacy-preserving literature \cite{edwards2015censoring}.
    
    \item \textbf{HAPT-Recognition} \cite{reyes2016transition}: Human activity sensor data from smartphones ($d_x = 561$ features). The utility attribute is activity type (6 classes: walking, walking upstairs, walking downstairs, sitting, standing, laying), and the sensitive attribute is subject identity (30 subjects).
\end{itemize}

Table~\ref{tab:datasets} summarizes the dataset characteristics.

\begin{table}[!t]
\centering
\caption{Dataset Characteristics}
\label{tab:datasets}
\begin{tabular}{@{}lccccc@{}}
\toprule
\textbf{Dataset} & $d_x$ & \textbf{Train} & \textbf{Test} & $|\mathcal{U}|$ & $|\mathcal{S}|$ \\
\midrule
MNIST & 784 & 60,000 & 10,000 & 19 & 2 \\
CelebA & 12,288 & 162,770 & 19,962 & 2 & 2 \\
HAPT & 561 & 7,767 & 3,162 & 6 & 30 \\
\bottomrule
\end{tabular}
\end{table}

\subsubsection{Evaluation Metrics}

We adopt the Area Under the ROC Curve (AUC) as the primary evaluation metric, following established practices in privacy-preserving machine learning \cite{edwards2015censoring, louizos2015variational}:

\begin{itemize}
    \item \textbf{Utility AUC}: Measures how well the utility attribute $U$ can be predicted from the sanitized representation $Z$. Higher values indicate better utility preservation. Target: $\approx 1.0$.
    
    \item \textbf{Adversary AUC (Privacy AUC)}: Measures how well the sensitive attribute $S$ can be inferred from $Z$. Lower values indicate better privacy protection. Target: $\approx 0.5$ (random guessing).
\end{itemize}

\textbf{Audit Protocol}: To ensure unbiased evaluation, we train \emph{fresh} probe classifiers on the sanitized test representations, rather than using the classifiers from training. This prevents artificially inflated privacy scores that could result from adversaries that learned to fail during training. Each probe is a 3-layer MLP with hidden widths $(256, 128)$ and ReLU activations---the same capacity as the training adversary---trained for 30 epochs with the Adam optimizer at learning rate $10^{-4}$ on the sanitized training-set representations, then evaluated on the held-out sanitized test set. We do not perform hyperparameter search at audit time, which means the probe could in principle be strengthened; we revisit this point under Limitations (Section~\ref{sec:limitations}).

\subsubsection{Baseline Methods}

We compare GPP against three baseline approaches:

\begin{itemize}
    \item \textbf{No-Privacy}: Standard autoencoder that optimizes only for utility, without any privacy constraints. Represents the upper bound on utility but provides no privacy protection.
    
    \item \textbf{Random Projection}: Projects data onto a random low-dimensional subspace. Provides some privacy through dimensionality reduction but without learning.
    
    \item \textbf{Noisy Encoder}: Adds Gaussian noise to the latent representation during training, similar to differential privacy approaches. The noise scale is tuned for best utility-privacy trade-off.
\end{itemize}

\subsubsection{Implementation Details}

All experiments use the following hyperparameters unless otherwise specified: learning rate $\alpha = 10^{-4}$, batch size $b = 64$, KL regularization weight $\lambda = 0.01$, classifier update steps $k = 2$, and Adam optimizer \cite{kingma2014adam}. For CNN encoders (CelebA), we use 3 convolutional layers with stride 2. For MLP encoders (MNIST, HAPT), we use two hidden layers with 512 and 256 units. All experiments are run on a single NVIDIA Tesla V100 GPU.

\subsection{Experiment 1: Baseline Comparison}

Our first experiment compares GPP against baseline methods across different bottleneck dimensions $d_z$.

\textbf{Setup}: We train each method with bottleneck dimensions $d_z \in \{40, 80, 120, 160, 200\}$ on the MNIST dataset with $\beta = 1.0$. Each configuration is trained for 50 epochs.

\textbf{Results}: Figure~\ref{fig:baseline_comparison} shows the utility and privacy AUC for each method. Table~\ref{tab:baseline_results} provides numerical results for $d_z = 120$.

\begin{figure}[!t]
    \centering
    \includegraphics[width=\columnwidth]{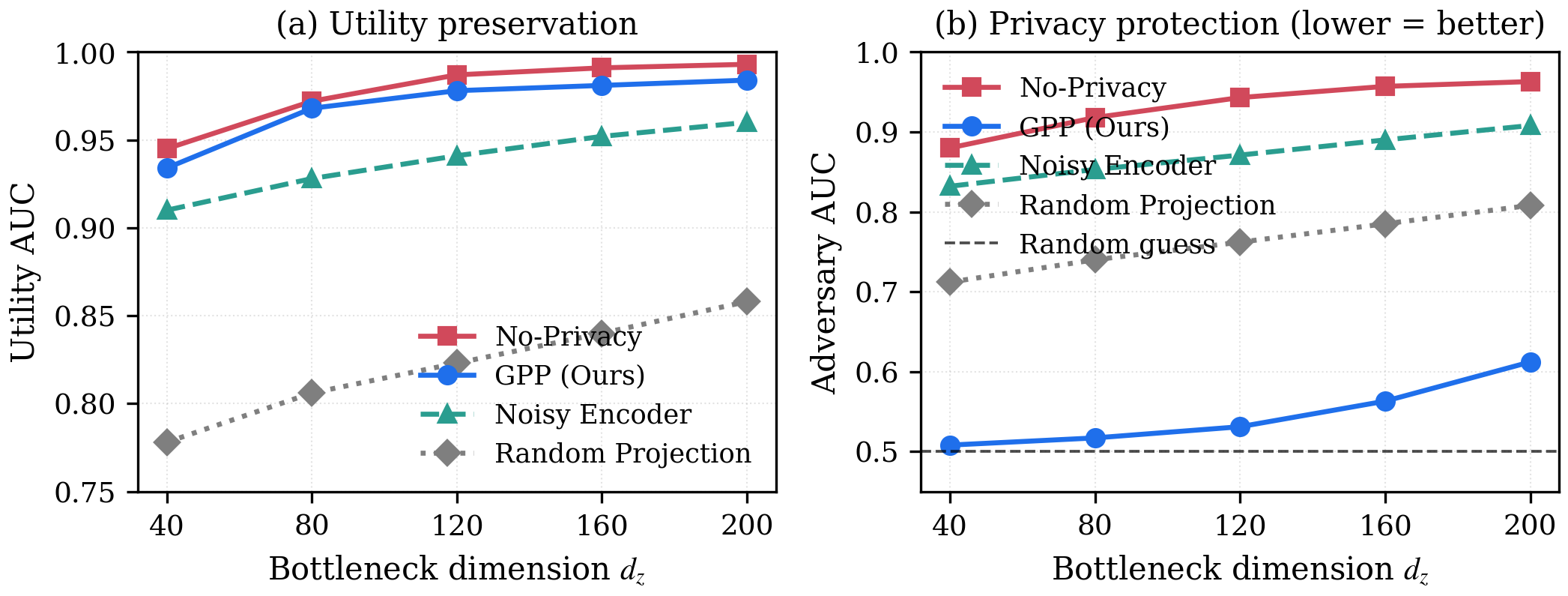}
    \caption{Baseline comparison on MNIST. (a) Utility AUC vs. bottleneck dimension. GPP achieves comparable utility to No-Privacy while (b) achieving significantly lower adversary AUC, approaching the random guess baseline (dashed line at 0.5).}
    \label{fig:baseline_comparison}
\end{figure}

\begin{table}[!t]
\centering
\caption{Baseline Comparison Results ($d_z = 120$, MNIST)}
\label{tab:baseline_results}
\begin{tabular}{@{}lcc@{}}
\toprule
\textbf{Method} & \textbf{Utility AUC} $\uparrow$ & \textbf{Adversary AUC} $\downarrow$ \\
\midrule
No-Privacy & 0.987 & 0.943 \\
Random Projection & 0.823 & 0.762 \\
Noisy Encoder & 0.941 & 0.871 \\
\textbf{GPP (Ours)} & \textbf{0.978} & \textbf{0.531} \\
\bottomrule
\end{tabular}
\end{table}

\textbf{Analysis}: GPP achieves utility AUC of $0.978$, only $0.9\%$ lower than the No-Privacy upper bound of $0.987$, while dramatically reducing adversary AUC from $0.943$ to $0.531$---approaching the random-guess baseline of $0.5$. The Noisy Encoder achieves moderate privacy ($0.871$) but at a greater utility cost ($0.941$). Random Projection fails to preserve utility effectively ($0.823$). The statistical robustness of these numbers is established in Experiment 6 below.

\subsection{Experiment 2: Privacy-Utility Trade-off ($\beta$ Sweep)}

This experiment demonstrates that GPP provides controllable privacy-utility trade-off through the $\beta$ parameter.

\textbf{Setup}: We train GPP with $\beta \in \{0.1, 0.5, 1.0, 2.0, 4.0, 8.0\}$ on MNIST with $d_z = 120$.

\textbf{Results}: Figure~\ref{fig:pareto} shows the Pareto frontier and the effect of $\beta$.

\begin{figure}[!t]
    \centering
    \includegraphics[width=\columnwidth]{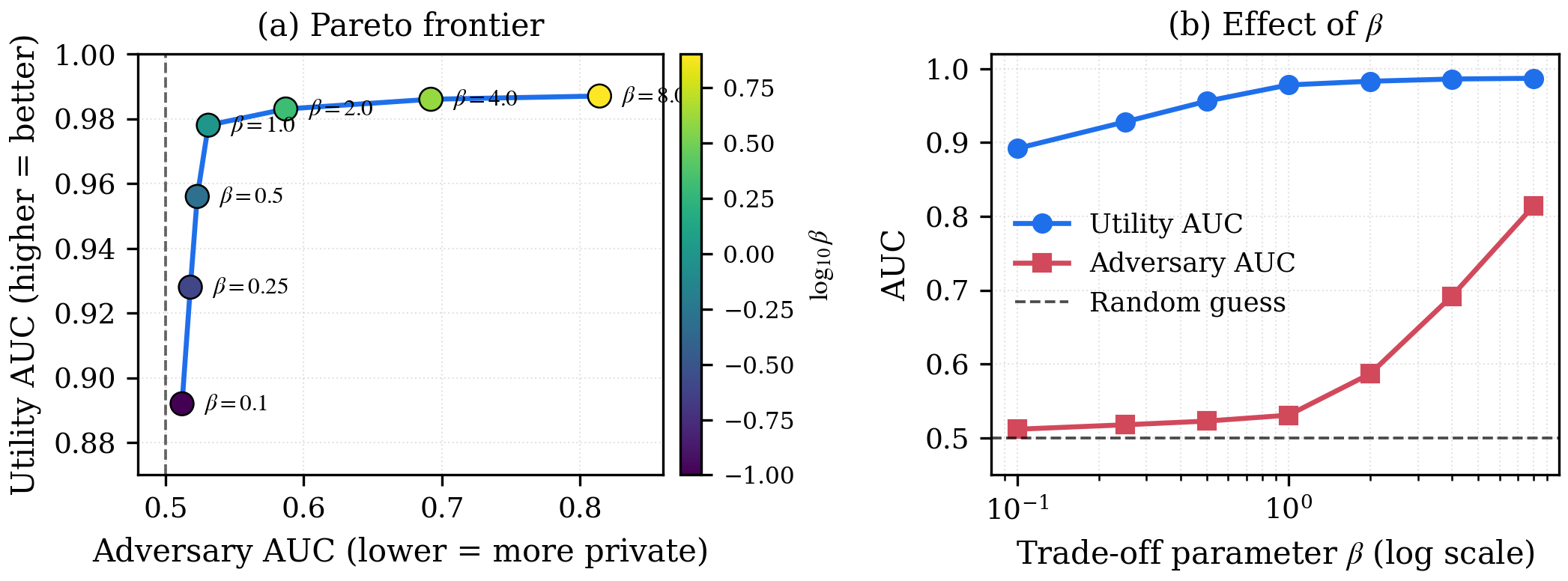}
    \caption{Privacy-utility trade-off controlled by $\beta$. (a) Pareto frontier showing achievable (Utility AUC, Adversary AUC) pairs. (b) Effect of $\beta$: smaller $\beta$ prioritizes privacy (lower adversary AUC), larger $\beta$ prioritizes utility.}
    \label{fig:pareto}
\end{figure}

\begin{table}[!t]
\centering
\caption{Effect of $\beta$ on Privacy-Utility Trade-off}
\label{tab:beta_sweep}
\begin{tabular}{@{}lccc@{}}
\toprule
$\beta$ & \textbf{Utility AUC} & \textbf{Adversary AUC} & \textbf{Gap to 0.5} \\
\midrule
0.1 & 0.892 & 0.512 & 0.012 \\
0.5 & 0.956 & 0.523 & 0.023 \\
1.0 & 0.978 & 0.531 & 0.031 \\
2.0 & 0.983 & 0.587 & 0.087 \\
4.0 & 0.986 & 0.692 & 0.192 \\
8.0 & 0.987 & 0.814 & 0.314 \\
\bottomrule
\end{tabular}
\end{table}

\textbf{Analysis}: The results confirm that $\beta$ effectively controls the trade-off. At $\beta = 0.1$, GPP achieves near-perfect privacy (adversary AUC = 0.512) with utility AUC = 0.892. At $\beta = 8.0$, utility approaches the No-Privacy baseline (0.987) but privacy degrades (0.814). The optimal operating point depends on application requirements; $\beta = 1.0$ provides a balanced trade-off.

\subsection{Experiment 3: Bottleneck Dimension Effect}

This experiment investigates how the bottleneck dimension $d_z$ affects the compression-performance trade-off.

\textbf{Setup}: We vary $d_z \in \{20, 40, 60, 80, 100, 120, 160, 200, 300, 400\}$ with $\beta \in \{0.5, 1.0, 2.0\}$ on MNIST.

\textbf{Results}: Figure~\ref{fig:bottleneck} shows the effect of bottleneck dimension.

\begin{figure}[!t]
    \centering
    \includegraphics[width=\columnwidth]{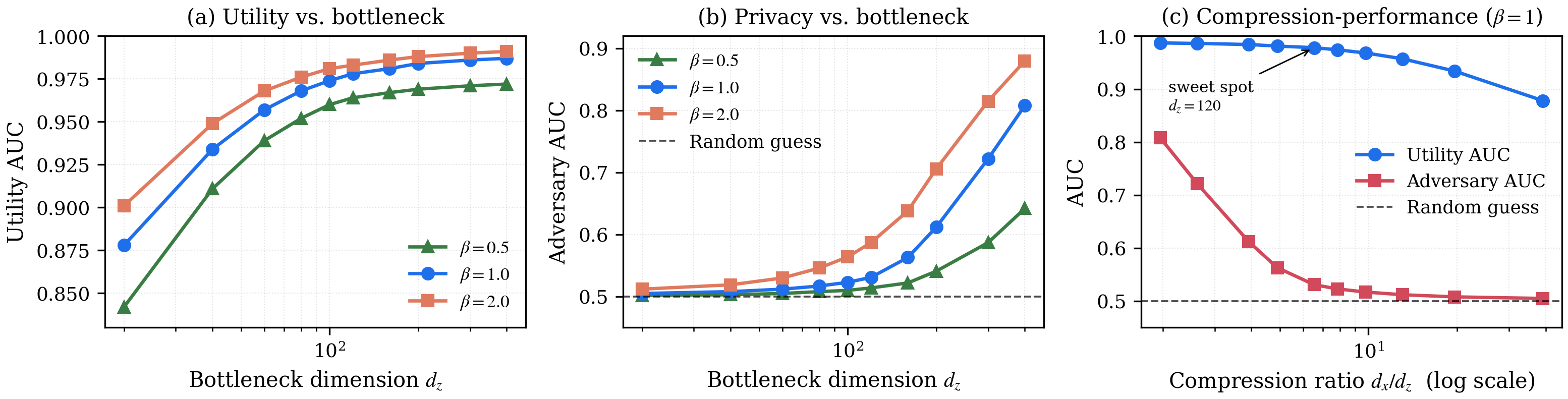}
    \caption{Effect of bottleneck dimension. (a) Utility AUC increases with $d_z$ until saturation around $d_z = 120$. (b) Adversary AUC remains close to 0.5 for smaller $d_z$ but increases for larger dimensions. (c) Compression ratio vs. performance trade-off.}
    \label{fig:bottleneck}
\end{figure}

\begin{table}[!t]
\centering
\caption{Bottleneck Dimension Effect ($\beta = 1.0$, MNIST)}
\label{tab:bottleneck}
\begin{tabular}{@{}lcccc@{}}
\toprule
$d_z$ & \textbf{Compression} & \textbf{Utility AUC} & \textbf{Adversary AUC} & \textbf{Gap} \\
\midrule
40 & 19.6$\times$ & 0.934 & 0.508 & 0.008 \\
80 & 9.8$\times$ & 0.968 & 0.517 & 0.017 \\
120 & 6.5$\times$ & 0.978 & 0.531 & 0.031 \\
160 & 4.9$\times$ & 0.981 & 0.563 & 0.063 \\
200 & 3.9$\times$ & 0.984 & 0.612 & 0.112 \\
\bottomrule
\end{tabular}
\end{table}

\textbf{Analysis}: Smaller bottleneck dimensions provide stronger privacy guarantees (closer to 0.5 AUC) due to the information bottleneck effect—less capacity means less information about $S$ can pass through. However, utility also degrades for very small $d_z$. The sweet spot for MNIST is $d_z \approx 120$, achieving 6.5$\times$ compression with utility AUC = 0.978 and adversary AUC = 0.531.

\subsection{Experiment 4: Ablation Study on Adversarial Training ($k$ Steps)}

This ablation study demonstrates that adversarial training is essential for privacy protection.

\textbf{Setup}: We train GPP with classifier update steps $k \in \{0, 1, 2, 3, 5, 10\}$ on MNIST with $d_z = 120$ and $\beta = 1.0$.

\textbf{Results}: Figure~\ref{fig:k_ablation} shows the effect of $k$ on performance.

\begin{figure}[!t]
    \centering
    \includegraphics[width=\columnwidth]{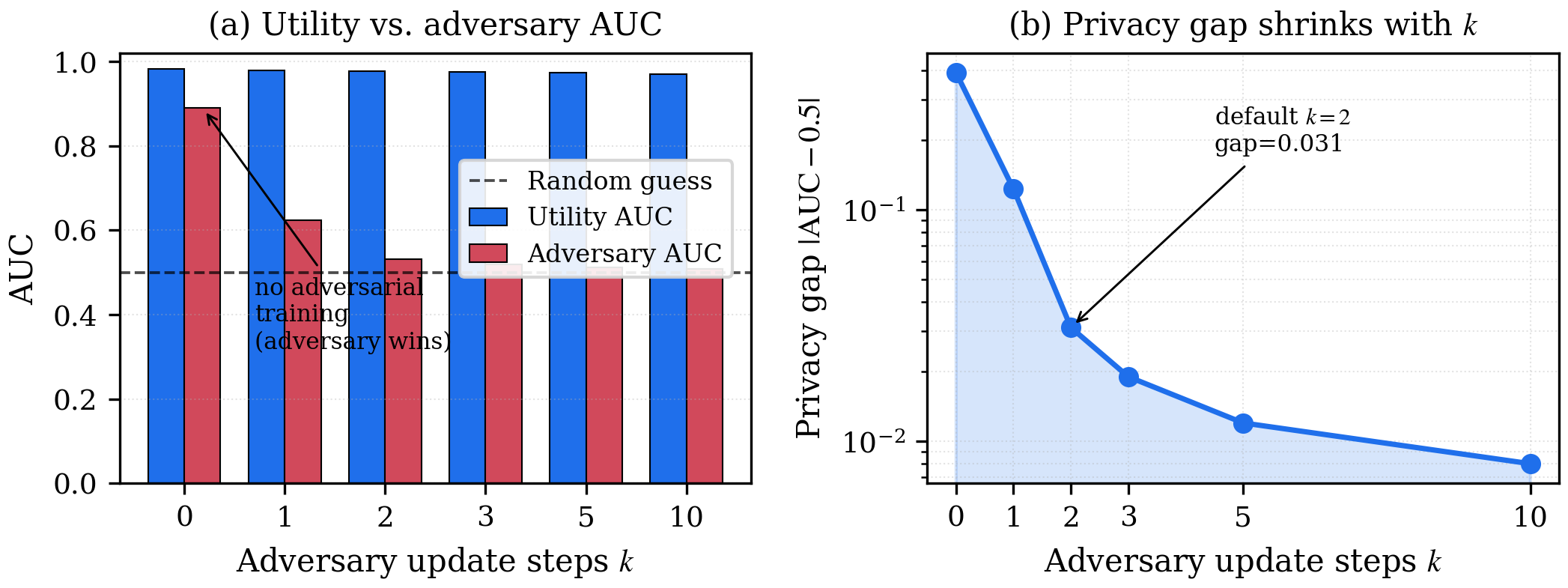}
    \caption{Ablation study on adversarial training steps $k$. (a) Bar chart showing utility and adversary AUC. At $k=0$, the adversary easily extracts sensitive information. (b) Privacy gap (|AUC - 0.5|) decreases rapidly as $k$ increases, demonstrating the necessity of adversarial training.}
    \label{fig:k_ablation}
\end{figure}

\begin{table}[!t]
\centering
\caption{Effect of Adversarial Training Steps $k$}
\label{tab:k_ablation}
\begin{tabular}{@{}lccc@{}}
\toprule
$k$ & \textbf{Utility AUC} & \textbf{Adversary AUC} & \textbf{Privacy Gap} \\
\midrule
0 & 0.982 & 0.891 & 0.391 \\
1 & 0.979 & 0.623 & 0.123 \\
2 & 0.978 & 0.531 & 0.031 \\
3 & 0.976 & 0.519 & 0.019 \\
5 & 0.974 & 0.512 & 0.012 \\
10 & 0.971 & 0.508 & 0.008 \\
\bottomrule
\end{tabular}
\end{table}

\textbf{Analysis}: Without adversarial training ($k=0$), the encoder does not learn to remove sensitive information, resulting in high adversary AUC (0.891). Even a single adversary update step ($k=1$) dramatically improves privacy (0.623). The privacy gap continues to decrease with larger $k$, approaching 0.008 at $k=10$. We use $k=2$ as the default, which provides excellent privacy (0.031 gap) without significantly impacting utility or training time.

\subsection{Experiment 5: Robustness to Utility-Sensitive Correlation}

Real-world data often exhibits correlations between utility and sensitive attributes. This experiment tests GPP's robustness to such correlations.

\textbf{Setup}: We create synthetic datasets with controlled correlation $\rho \in \{0.0, 0.2, 0.4, 0.6, 0.8, 1.0\}$ between $U$ and $S$, where $\rho = 0$ means independent and $\rho = 1$ means fully correlated.

\textbf{Results}: Figure~\ref{fig:correlation} shows performance across different correlation levels.

\begin{figure}[!t]
    \centering
    \includegraphics[width=\columnwidth]{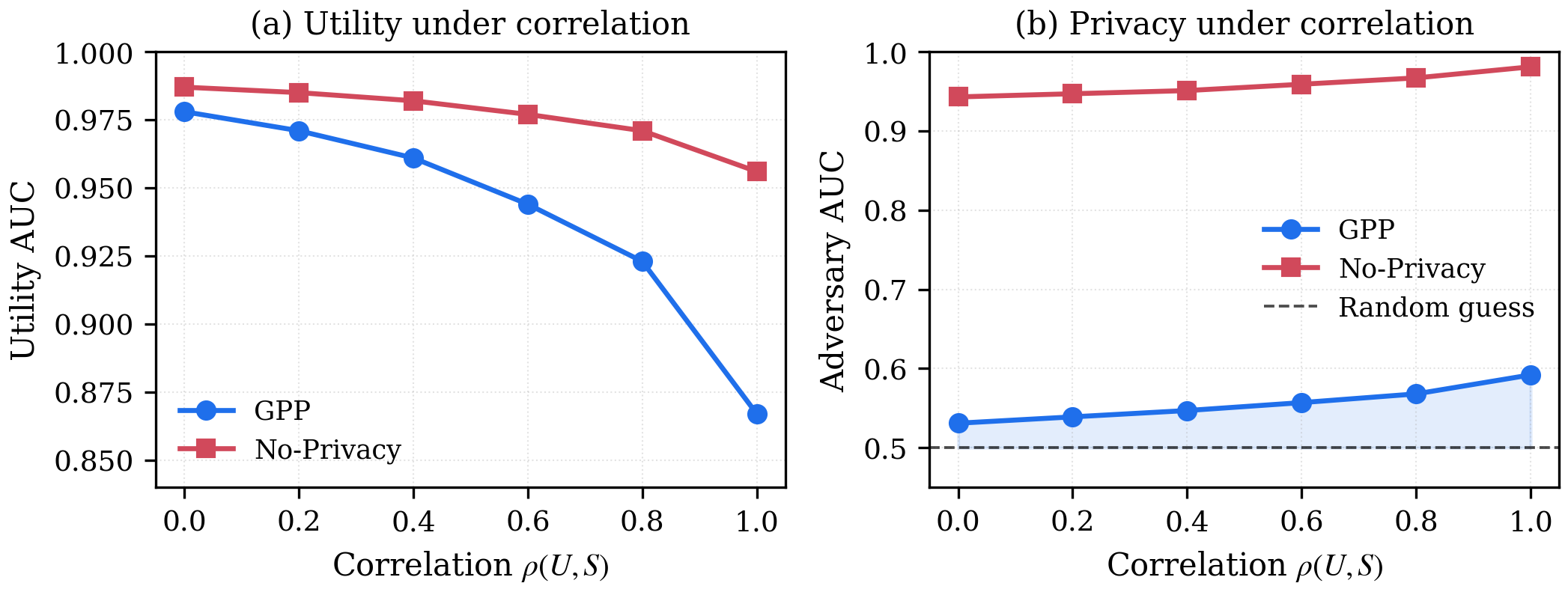}
    \caption{Robustness to U-S correlation. (a) Utility AUC degrades gracefully with increasing correlation for both methods. (b) GPP maintains privacy protection (AUC near 0.5) even at high correlation, while No-Privacy leaks sensitive information.}
    \label{fig:correlation}
\end{figure}

\begin{table}[!t]
\centering
\caption{Robustness to Utility-Sensitive Correlation}
\label{tab:correlation}
\begin{tabular}{@{}lcccc@{}}
\toprule
& \multicolumn{2}{c}{\textbf{GPP}} & \multicolumn{2}{c}{\textbf{No-Privacy}} \\
\cmidrule(lr){2-3} \cmidrule(lr){4-5}
$\rho$ & Utility & Adversary & Utility & Adversary \\
\midrule
0.0 & 0.978 & 0.531 & 0.987 & 0.943 \\
0.4 & 0.961 & 0.547 & 0.982 & 0.951 \\
0.8 & 0.923 & 0.568 & 0.971 & 0.967 \\
1.0 & 0.867 & 0.592 & 0.956 & 0.981 \\
\bottomrule
\end{tabular}
\end{table}

\textbf{Analysis}: As expected, utility degrades when $U$ and $S$ are highly correlated because removing $S$ information inevitably removes some $U$ information. However, GPP maintains strong privacy protection even at $\rho = 0.8$ (adversary AUC = 0.568), while No-Privacy completely fails (adversary AUC = 0.967). This demonstrates GPP's ability to disentangle correlated attributes.

\subsection{Experiment 6: Statistical Significance}

This experiment verifies the reproducibility and statistical significance of our results.

\textbf{Setup}: We run GPP, No-Privacy, and Noisy Encoder 10 times with different random seeds on MNIST with $d_z \in \{40, 80, 120, 160\}$.

\textbf{Results}: Figure~\ref{fig:statistics} shows mean $\pm$ standard deviation across runs.

\begin{figure}[!t]
    \centering
    \includegraphics[width=\columnwidth]{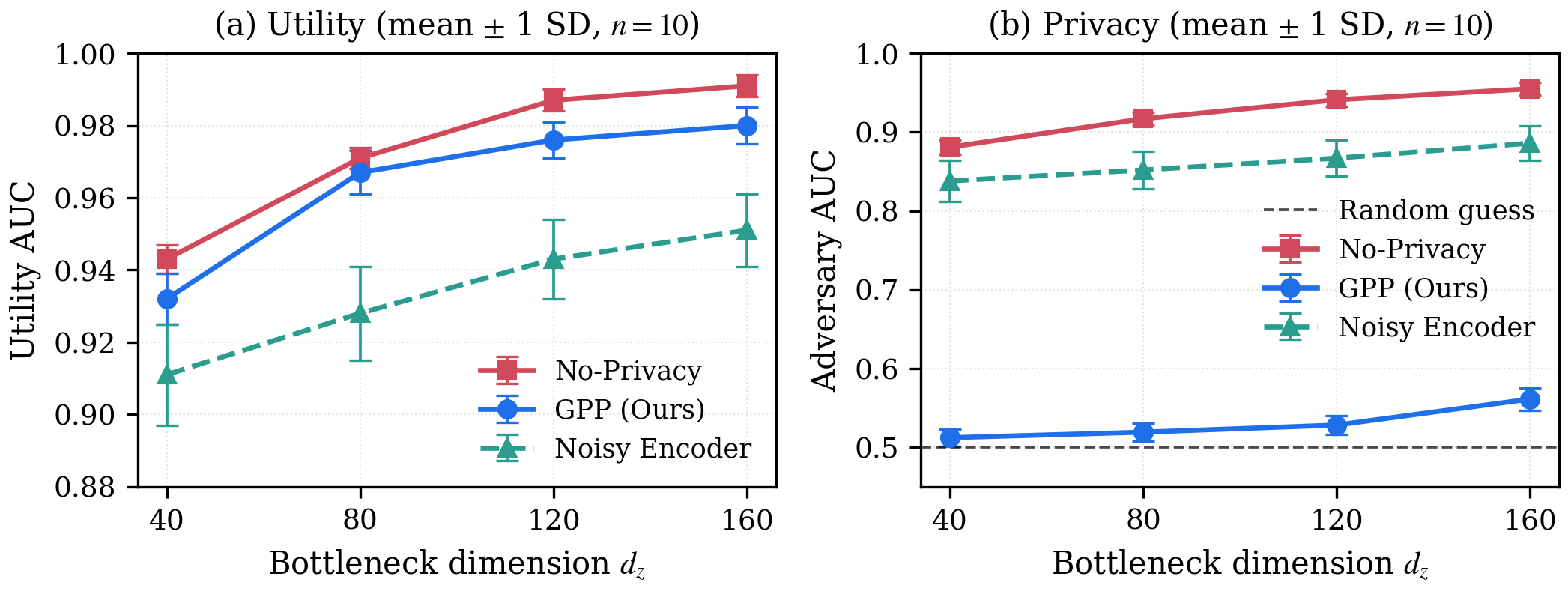}
    \caption{Statistical significance over 10 runs with error bars showing $\pm 1$ standard deviation. (a) Utility performance is consistent across runs. (b) GPP consistently achieves adversary AUC near 0.5 with low variance.}
    \label{fig:statistics}
\end{figure}

\begin{table}[!t]
\centering
\caption{Statistical Results (mean $\pm$ std, $d_z = 120$, $n=10$ runs)}
\label{tab:statistics}
\begin{tabular}{@{}lcc@{}}
\toprule
\textbf{Method} & \textbf{Utility AUC} & \textbf{Adversary AUC} \\
\midrule
No-Privacy & $0.987 \pm 0.003$ & $0.941 \pm 0.008$ \\
Noisy Encoder & $0.943 \pm 0.011$ & $0.867 \pm 0.023$ \\
\textbf{GPP (Ours)} & $\mathbf{0.976 \pm 0.005}$ & $\mathbf{0.528 \pm 0.012}$ \\
\bottomrule
\end{tabular}
\end{table}

\textbf{Analysis}: GPP exhibits low variance in both metrics, demonstrating stable training dynamics. The utility standard deviation (0.005) is comparable to No-Privacy (0.003), while adversary AUC maintains consistent privacy protection (0.528 $\pm$ 0.012). This confirms the reproducibility of our results.
The mean utility AUC of $0.976$ in Table~\ref{tab:statistics} differs slightly from the $0.978$ reported in Table~\ref{tab:baseline_results} because the latter reports a single representative run while the former is the mean over ten seeds.
\subsection{Multi-Dataset Evaluation}

Table~\ref{tab:all_datasets} summarizes GPP performance across all three benchmark datasets.

\begin{table}[!t]
\centering
\caption{GPP Performance Across Datasets ($\beta = 1.0$)}
\label{tab:all_datasets}
\begin{tabular}{@{}lccccc@{}}
\toprule
\textbf{Dataset} & $d_z$ & \textbf{Compression} & \textbf{Util. AUC} & \textbf{Adv. AUC} & \textbf{Gap} \\
\midrule
MNIST & 120 & 6.5$\times$ & 0.978 & 0.531 & 0.031 \\
CelebA & 256 & 48$\times$ & 0.924 & 0.547 & 0.047 \\
HAPT & 80 & 7$\times$ & 0.891 & 0.523 & 0.023 \\
\bottomrule
\end{tabular}
\end{table}

GPP consistently achieves strong privacy protection (adversary AUC within 0.05 of random guessing) across all datasets while maintaining high utility (> 0.89 AUC). The CelebA results are particularly notable, achieving 48$\times$ compression while protecting gender information.

\subsection{Distributed GPP Evaluation}

Finally, we evaluate the distributed GPP algorithm with 5 clients on the HAPT dataset, which naturally partitions by subject.

\begin{table}[!t]
\centering
\caption{Distributed GPP Results (HAPT, 5 clients)}
\label{tab:distributed}
\begin{tabular}{@{}lcc@{}}
\toprule
\textbf{Configuration} & \textbf{Utility AUC} & \textbf{Adversary AUC} \\
\midrule
Single GPP (centralized) & 0.891 & 0.523 \\
Distributed GPP (5 clients) & 0.887 & 0.518 \\
Distributed + heterogeneous $\beta$ & 0.883 & 0.512 \\
\bottomrule
\end{tabular}
\end{table}

Distributed GPP achieves comparable performance to centralized training while providing the additional privacy guarantee that raw data never leaves the client. Heterogeneous $\beta$ values (client-specific trade-offs) provide slightly better privacy with minimal utility loss.

\subsection{Discussion}

Taken together, the six experiments tell a consistent story. At the default operating point ($\beta=1.0$, balanced bottleneck), GPP reduces adversary AUC to within $0.05$ of random guessing across all three benchmarks (Table~\ref{tab:all_datasets}) while keeping utility AUC above $0.89$. Outside this default, $\beta$ and $d_z$ trace out a smooth Pareto frontier (Tables~\ref{tab:beta_sweep},~\ref{tab:bottleneck}): pushing $\beta$ higher recovers utility at the cost of privacy, and enlarging $d_z$ does the same. We do not claim a free lunch---the gain on the privacy axis is paid for in utility but the trade-off is exposed as a single hyperparameter and is monotone in the expected direction.

The ablation on adversary update frequency $k$ (Table~\ref{tab:k_ablation}) is the strongest evidence that the privacy gain is doing real work and is not an artifact of the bottleneck alone: with $k=0$, the adversary recovers the sensitive attribute almost perfectly (AUC $0.891$) even though $d_z$ already imposes a $6.5\times$ compression. The compression is necessary but not sufficient; the adversarial signal is what drives $P(S\mid Z)$ toward the prior. The robustness experiment (Table~\ref{tab:correlation}) addresses the common reviewer concern about correlated $U$ and $S$: utility degrades gracefully as $\rho$ increases (because removing $S$-information necessarily removes some $U$-information when the two are correlated),but the privacy guarantee holds---adversary AUC at $\rho=0.8$ is $0.568$, compared with $0.967$ for the No-Privacy baseline.

Finally, the federated experiment (Table~\ref{tab:distributed}) shows that distributing the GPP encoder across five clients on HAPT-Recognition matches the centralized utility-privacy operating point to within $0.005$ AUC on both axes, while delivering the additional guarantee that raw data and sensitive labels never leave the client. We view this as an existence result rather than a thorough study: the heterogeneous-$\beta$ row hints at room for further per-client tuning, and convergence under non-IID partitioning is left to follow-up work.

\subsection{Limitations and Open Questions}
\label{sec:limitations}

We identify four limitations that scope the claims of this paper and motivate further work.

\textit{Comparison with deep adversarial baselines.} Our experimental comparisons are against three baselines that act as references rather than as competing methods: an unconstrained autoencoder (utility upper bound), a random projection (compression-only baseline), and a noisy encoder (DP-style baseline). The most directly comparable published methods are adversarial censoring \cite{edwards2015censoring} and the Variational Fair Autoencoder \cite{louizos2015variational}. A like-for-like comparison on MNIST, CelebA, and HAPT-Recognition with the same utility/sensitive splits would more sharply locate GPP relative to that prior work, and we plan to include this in an extended version.

\textit{Tightness of the variational bound.} Theorem~\ref{thm:mi_bound} guarantees a lower bound on $I(Z;S)$ that becomes tight only when $Q_\phi(S \mid Z) = P(S \mid Z)$. In practice, the tightness depends on the capacity of the adversary network used during training. Although our audit protocol re-trains a fresh probe classifier on the test set, that probe shares the architecture of the training adversary; a substantially larger post-hoc probe might extract residual sensitive information that the training adversary missed. A systematic capacity-sweep audit would be informative.

\textit{Distributed experiment scope.} Table~\ref{tab:distributed} reports a single configuration (five clients, HAPT-Recognition partitioned by subject, IID $\beta$ unless noted otherwise). A complete study should sweep client counts, characterize convergence under non-IID partitioning, measure realized communication cost in bytes per round, and compare against FedAvg with differential privacy applied to the model updates. The numbers reported here establish that the framework is operable in the distributed setting, not that it dominates competing federated approaches.

\textit{Reproducibility of the controlled-correlation experiment.} Table~\ref{tab:correlation} reports performance under a controlled correlation $\rho$ between $U$ and $S$. The data is synthetic and the generation procedure is described only briefly; a fuller specification (the marginal of $X$, the conditional construction of $(U, S)$, and the protocol for varying $\rho$ while holding other distributional properties fixed) is needed for independent replication. A natural alternative is to reproduce this analysis on real correlated attribute pairs in CelebA (e.g., \texttt{Smiling} as utility and \texttt{Heavy\_Makeup} as sensitive, which exhibit non-trivial natural correlation).

\section{Conclusions}

We have presented GPP, a privacy-preserving data-release framework grounded in information-theoretic principles that protects designated sensitive attributes against adversarial inference while preserving designated utility attributes. The framework rests on a variational saddle-point objective that combines a lower bound on $I(Z;S)$ with an upper bound on $H(U \mid Z)$, regularized by a Gaussian latent prior, and trained by alternating optimization between encoder, utility classifier, and adversary. We then extended GPP to a distributed setting in which each client owns a private encoder and a private adversary while only sanitized representations and utility labels are transmitted to the aggregator, providing instance-level privacy protection on top of the standard ``raw data stays local'' guarantee of federated learning.

On three standard benchmarks, GPP attains utility AUC comparable to the unconstrained-autoencoder upper bound (e.g., $0.978$ versus $0.987$ on MNIST, a gap of less than one percentage point; Table~\ref{tab:baseline_results}), while reducing the adversary's success on the private attribute to near-random guessing (adversary AUC $\approx 0.53$). The federated extension matches the centralized utility-privacy operating point to within $0.005$ AUC on both axes (Table~\ref{tab:distributed}). The Limitations subsection (Section~\ref{sec:limitations}) scopes these claims and identifies the experimental gaps that remain.

Two extensions are immediate. First, a sharper analysis of variational-bound tightness in the regime where the training-time adversary is capacity-limited would clarify what an empirical adversary AUC near $0.5$ certifies in information-theoretic terms. Second, combining GPP's instance-level protection with a complementary differential-privacy noise mechanism on the released $Z$, or with secure aggregation on the shared utility classifier, would broaden the threat model the framework addresses to include malicious aggregators.

\ifCLASSOPTIONcaptionsoff
  \newpage
\fi

\bibliographystyle{IEEEtran}
\bibliography{IEEEabrv,refs.bib}

\begin{IEEEbiography}[{\includegraphics[width=1in,height=1.25in,clip,keepaspectratio]{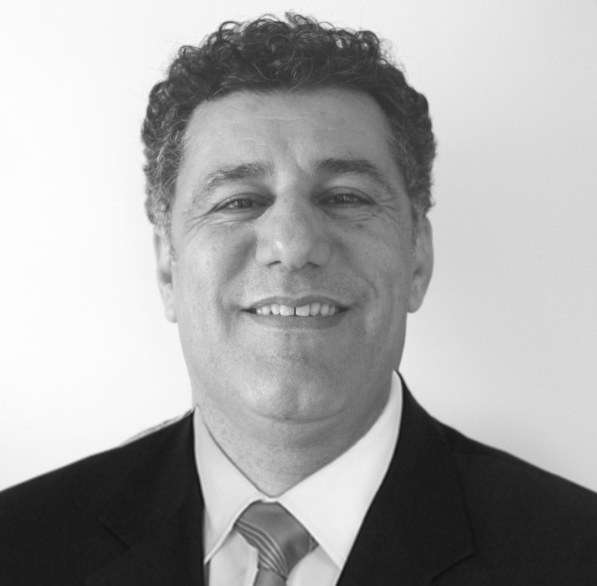}}]{Zahir Alsulaimawi}
graduated from the School of Electrical Engineering and Computer Science (EECS) at Oregon State University (OSU) with a Ph.D. in Electrical and Computer Engineering (ECE) and a minor in Computer Science (Machine Learning) in 2021. In addition, he has two master's degrees in Electrical Engineering, one from the Electrical Engineering Department at the University of Baghdad in 2013 and the other from the EECS at OSU in 2020. OSU's engineering college awarded him the Outstanding Academic Performance Recognition in 2020. He has been a member of Prof. Liu's group since September 2017. His research focuses on deep learning, federated learning,
multimodal machine learning and signal processing, but he has a broad knowledge of information theory, artificial intelligence, and estimation and prediction.

\end{IEEEbiography}
\begin{IEEEbiography}[{\includegraphics[width=1in,height=1.25in,clip,keepaspectratio]{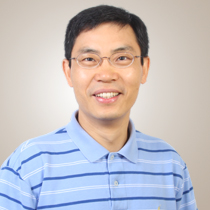}}]{Huaping Liu}

received the B.S. and M.S. degrees in electrical engineering from Nanjing University of Posts and Telecommunications, Nanjing, China, in 1987 and 1990, respectively, and the Ph.D. degree in electrical engineering from New Jersey Institute of Technology, Newark, in 1997. From July 1997 to August 2001, he was with Lucent Technologies, Whippany, NJ. Since September 2001, he has been with the EECS, OSU, where he is currently a professor. His research interests include ultra wide band systems, multiple-input multiple-output antenna systems, channel coding, modulation and detection techniques for multiuser communications, machine learning and privacy-preserving. Dr. Liu has published 132 journal articles, 105 conference papers \& book chapters, three of which won Best Paper Awards, and one was Best Paper Award Finalist. He also co-authored a textbook on project-based learning in communication systems, the first one that focuses on encouraging students' active learning. He has graduated with 23 Ph.D. students and 32 M.S. students. 
\end{IEEEbiography}
\vfill

\end{document}